\documentclass[10pt,journal,compsoc]{IEEEtran}
\ifCLASSOPTIONcompsoc
  \usepackage[nocompress]{cite}
\else
  \usepackage{cite}
\fi

\ifCLASSINFOpdf
\else
\fi

\usepackage{algorithm}
\usepackage{algpseudocode}
\usepackage{graphicx}
\usepackage{enumerate}
\usepackage{subcaption}
\usepackage{blindtext}
\usepackage{amsmath}
\usepackage{pbox}
\usepackage[outdir=Figures/]{epstopdf}
\usepackage{enumitem} 
\usepackage{tabularx}

\newtheorem{theorem}{\bf Theorem}

\newtheorem{lemma}[theorem]{\bf Lemma}

\def \sys {\textit{DynamicSLAM}}
\usepackage{color}

\begin{document}
\title{DynamicSLAM: Leveraging Human Anchors for Ubiquitous Low-Overhead Indoor Localization}

\author{Ahmed~Shokry,~\IEEEmembership{Student Member,~IEEE,}
        Moustafa~Elhamshary,~\IEEEmembership{Member,~IEEE}
        and~Moustafa~Youssef,~\IEEEmembership{Fellow,~IEEE}
        
\IEEEcompsocitemizethanks{
\IEEEcompsocthanksitem  Ahmed Shokry is with the Department of Computer Science, Alexandria University, Egypt.\protect\\
E-mail: ahmed.shokry@alexu.edu.eg

\IEEEcompsocthanksitem  Moustafa Elhamshary and Moustafa Youssef are with Egypt-Japan University for Science and Technology (E-JUST), Alexandria, Egypt.\protect\\
E-mail: {moustafa.elhamshary,moustafa.youssef}@ejust.edu.eg
}
\thanks{Under submission December 26, 2018.}}

\markboth{IEEE TRANSACTIONS ON MOBILE COMPUTING, VOL. X, NO. X, XXXX 2020}%
{Shell \MakeLowercase{\textit{et al.}}: Bare Demo of IEEEtran.cls for Computer Society Journals}
\IEEEtitleabstractindextext{%
\begin{abstract}
 We present \sys{}: an  indoor localization technique that eliminates the need for the daunting calibration step. 
\sys{} is a novel Simultaneous Localization And Mapping  (SLAM) framework that iteratively acquires the  feature map of the environment while simultaneously localizing users relative to this map. Specifically, we employ the phone inertial sensors to keep track of the user's path. To compensate for the error accumulation due to the low-cost inertial sensors, \sys{} leverages unique  points in the environment (anchors) as observations to reduce the estimated location error.  \sys{} introduces the novel concept of mobile human anchors that are based on the encounters with  other users in the environment, significantly increasing the number and ubiquity of anchors and boosting localization accuracy. 
 We present different encounter models and show how they are incorporated in a unified probabilistic framework to reduce the ambiguity in the user location.
Furthermore, we present a theoretical proof for system convergence and the human anchors ability to reset the accumulated error.

Evaluation of \sys{} using different Android phones shows that it can provide a localization accuracy with a median of 1.1m. This accuracy  outperforms the state-of-the-art techniques by 55\%, highlighting \sys{} promise for ubiquitous indoor localization.
\end{abstract}

\begin{IEEEkeywords}
Unconventional localization, SLAM, indoor localization, unsupervised localization.
\end{IEEEkeywords}}

\maketitle

\IEEEdisplaynontitleabstractindextext
\IEEEpeerreviewmaketitle

\graphicspath{{./Figures/}}
\IEEEraisesectionheading{
\section{Introduction}}
Recent years have witnessed the advent of indoor localization systems
harnessing the many untapped capabilities of phones-embedded sensors  and their 
widespread availability. 
Most developed  indoor localization techniques are WiFi-based, leveraging
the Received Signal Strength (RSS) from WiFi access
points (APs) as the metric for the location determination. To combat the noisy wireless channel, typically a WiFi fingerprint  map is  built that captures the signatures of the APs at different locations in the 
area of interest. Nonetheless, the deployment
cost is prohibitive as this requires a  calibration phase; which is time consuming,
labour intensive and is vulnerable to environmental dynamics.
To tackle this problem, a number of approaches for automating
the fingerprinting process have been proposed including the use of RF
propagation models \cite{elbakly2016robust}, combining RF localization with other
sensors \cite{wang2012no}, or crowdsourcing the fingerprint; where users perform
the required survey process in the real-time  \cite{rai2012zee,park2010growing}. These techniques, however, suffer from  low location accuracy and/or
require explicit user intervention. To further address these issues,
some algorithms, e.g., \cite{jin2011robust}, leverage the phone inertial sensors to track  the user starting from a reference point. However, the localization error quickly accumulates due to the noise of the low-cost inertial sensors. To alleviate this problem, researchers, e.g. \cite{rai2012zee}, proposed to match the user path with the map information to remove any outlier locations. 
Nonetheless, these indoor localization techniques require the knowledge of  indoor map (i.e., wall and hallway locations).

Recently, Simultaneous Localization and Mapping (SLAM)  approaches, commonly used in the robotics domain, have been introduced to provide iterative and  autonomous feature maps (i.e.,  anchor positions and signatures; not the indoor floorplan) of previously unknown environments and localization of a subject within this environment simultaneously. For instance,  ActionSLAM \cite{hardegger2012actionslam} leverages foot-mounted inertial sensors to track the user's  indoor location  using detected user actions as observations.  However, foot-mounted sensors are not ubiquitous and the recognized activities  cannot be accurately detected by standard smartphone sensors, limiting this technique scalability. 
 To extend SLAM  to standard smartphones, SemanticSLAM \cite{abdelnasser2016semanticslam} leverages phone inertial sensors to continuously estimate the user's location and  detect supervised and unsupervised anchors in the building structure (e.g., stairs, elevators, etc.) with  unique patterns that surface on the phone sensors when passing through them. These anchors provide error resetting opportunities to control the inertial sensors error accumulation. Although SemanticSLAM achieves good localization accuracy, its performance is contingent on the  anchor density \cite{abdelnasser2016semanticslam}, which may be limited in many scenarios when the user moves in only one floor or not passing by any elevation-change anchor; losing many opportunities for  error-resetting.

In this paper, we propose \sys{}, a novel SLAM framework that can provide accurate, robust, low-overhead,
and ubiquitous indoor localization. \sys{} is a plug-and-play indoor localization system, where no infrastructure support nor war driving  are needed for the operation of the system; making it  building-independent. As in traditional SLAM, \sys{} tracks the user's location using dead reckoning based on the phone inertial sensors and leverages \textit{stationary anchors} to curb the accumulated location error. However, \textit{to further increase the anchor density and hence the localization accuracy in any environment}, \sys{} introduces the novel concept of \textbf{\textit{human mobile anchors}}. Specifically, \sys{}  leverages humans walking in the environment as anchors for error-resetting opportunities.
To achieve this, \sys{}  needs to address a number of challenges, including recognizing users' encounters in a ubiquitous manner, determining  their relative location, and handling the inherent error in the users' estimated locations. 

We evaluate the performance of \sys{} in a real deployment based on different Android phones.  Our results show that it can achieve a median localization error of 1.1m, outperforming  other state-of-the-art techniques by more than 55\%.
In summary, we provide the following contributions in this paper:
\begin{itemize}
    \item \textbf{We present a novel collaborative simultaneous localization and mapping framework that leverages other users in the environment as mobile anchors to enhance the localization and mapping accuracy in environments with low building anchor density.} 
    \item \textbf{We present different ubiquitous observation models to  detect users' encounters and estimate  their relative  distances.}
    \item \textbf{We present a theoretical proof for system convergence and the human anchors ability to reset the accumulated error.}
     \item \textbf{We implement and evaluate our system on different Android phones and quantify its performance relative to the state-of-the-art systems.}
\end{itemize}
The rest of the paper is organized as follows:  Section~\ref{sec:background} gives a brief background about the traditional SLAM framework.
Section \ref{sec:over} gives an overview of the proposed system 
 main components while Section~\ref{sec:social} gives the details of the \sys{} probabilistic framework.
 We evaluate the system
performance in Section~\ref{sec:eval}. 
Finally sections~\ref{sec:related} and \ref{sec:conc} discuss related work and conclude the paper respectively.
\section{SLAM Background}
\label{sec:background} 
Simultaneous Localization and Mapping (SLAM) is a computational problem  originally introduced in  robotics where a mobile robot is  placed
at an unknown location in an unknown environment and 
the robot  continuously localizes itself within this environment,
 while simultaneously incrementally building a consistent \textit{feature} map of this environment   \cite{durrant2006simultaneous}. The points constituting the feature map are also called anchors \cite{abdelnasser2016semanticslam}. In SLAM,  the trajectory of the robot and the locations and signatures of all anchors are estimated on-line  without the need for any a priori knowledge. 
To do that, SLAM harnesses the robot  odometers  to keep track of
the robot motion (i.e., to estimate its pose)  as well as  its exteroceptive sensors for detecting anchors in the environment to build the map.  In a probabilistic form, the goal of SLAM is to find the pose $s^t$ and map $\Theta$ that maximize the  joint  probability density function:
\begin{equation}
\label{eq:goal}
    p({s}_t, \Theta|u^t, z^t,n^t) 
\end{equation}

where  $u^t = \{u_1,..,u_t\}$   is the  set  of motion update measurements (displacement and heading change) capturing the control history till time $t$, $z^t = \{z_1,..,z_t\}$ is the history of anchor location observations till time $t$, and $n^t = \{n_1,...,n_t\}$ is  history of anchor identifications (data associations) where $n_t$
specifies the identity of the anchor observed at time $t$.  

The earliest SLAM algorithm \cite{durrant2006simultaneous}
is based on the extended Kalman filter (EKF). The EKF
represents the robot's map and position by a
high-dimensional Gaussian density over all map anchors
and the robot pose. The off-diagonal elements
in the covariance matrix of this multivariate Gaussian
represent the correlation between errors in the robot
pose and the anchors in the map. Therefore, the
EKF can accommodate the natural correlation of errors
in the map. To  estimate the joint probability
density function in Equation \ref{eq:goal}, the EKF-SLAM algorithm factors it into two different models:   the state transition model (motion update model) $p({s_t}|u_t, {s_{t-1}})$ and the observation model $p(z_t|s_t, \Theta_{n_t}, n_t)$,
 describing the effect of the control input and
observation respectively. 
 The motion update model is assumed to be a Markov
process in which the next position depends only on the
immediately preceding position and the applied control
$u_k$, and is independent of both the observations and the
map. The observation model describes the probability of
making an observation $z_k$ when the robot location and
anchor locations are known. 
EKF-SLAM  approach has several drawbacks. First,  its 
computation complexity is quadratic in the
number of anchors \cite{thrun2005probabilistic}.  Second, EKF assumes a Gaussian noise  for
 the robot motion and observation. In this case, the amount of uncertainty in sensors must be relatively small, since the linearization in EKF may introduce intolerable noise.
  Third, EKF represents posteriors by a Gaussian, which is a unimodal distribution that works well only in tracking problems where the posterior is focused around the true state with a
small margin of uncertainty.  Gaussian posteriors are a poor match for many
global estimation problems in which many distinct hypotheses exist, each of
which forms its own mode in the posterior. This directly impacts 
the data association, i.e., how to determine the
identity of the detected anchors when multiple of
them have a similar signature (e.g. two nearby stairs), which can lead to different maps
based on the chosen association.  

FastSLAM \cite{montemerlo2002fastslam,michael2003fastslam,montemerlo2007fastslam} is an alternative SLAM solution that  mitigates the limitation of EKF-SLAM.
 FastSLAM combines particle filter and extended
Kalman filter. 
The basic idea of FastSLAM is to maintain a set of particles. Each particle
contains a sampled robot path and its own representation of the  map  where each anchor
in the map is represented by its own local Gaussian instead of
one big joint Gaussian as  in the case of EKF-SLAM algorithm. This is  based on the SLAM property where the map anchors are conditionally independent
given the path. By factoring out the path (one per particle), we can 
simply treat each map anchor  independently, thereby avoiding the costly
step of maintaining the correlations between them that plagued the EKF-SLAM
approach.  The resulting
representation requires space linear in the size of the map, and linear in
the number of particles. Estimating the posterior over all
paths enables the factorization of posterior over poses and maps:
\begin{equation}
p(s^t,\Theta |z^t, u^t, n^t) = p( s^t|z^t, u^t, n^t) \prod_{n=1}^{N_L}\  p(\Theta_n|s^t, z^t, n^t)
\end{equation}
This factorization states that the calculation of the posterior over paths and
maps can be decomposed into $N_L+1$ probabilities, where $N_L$  is the number of anchors. FastSLAM estimates the robot path $s^t$ using a particle filter (Localization problem).  The second term (Mapping problem) can be factored  into $N_L$ separate problems, 
 where the individual anchor location probability density function $p(\theta_n|s^t, z^t, n^t)$ is estimated using an EKF. 
Formally, the posterior of the $m^{th}$ particle ($S_t^{[m]}$)
contains a path $s^{t,[m]}$ and $N_L$  local estimates of  anchors
$\Theta^{[m]}$. A local  estimate of anchor $n$ by particle $[m]$ is represented by anchor type$\hat {f}_{n,t}$, the mean of anchor position $\mu_{n,t}^{[m]}$, and the position covariance $\sum_{n,t}^m$.

FastSLAM has many  advantages  over  EKF-SLAM. First, FastSLAM can be implemented in time logarithmic in the number
of anchors \cite{thrun2005probabilistic}, offering computational advantages over the plain
EKF-SLAM. Second, FastSLAM provides  significant 
robustness in data association  as   it maintains posteriors
over multiple data associations, 
and not just by tracking only a single
data association at any point in time based on the incremental
maximum likelihood data association incorporated in EKF-SLAM \cite{montemerlo2002fastslam}. Finally, FastSLAM can also cope with
non-linear models as particle filters can handle non-linear robot motion models and is guaranteed to converge under
certain assumptions \cite{montemerlo2002fastslam}. 
Due to these
advantages, we adopt FastSLAM approach in \sys{}.
\section{System Overview}
\label{sec:over}
Now, we start with an overview of \sys{}  core building blocks (Figure \ref{overview}), postponing the discussion on the core SLAM algorithm to the next section.

\subsection{ User Mobile Traces}
\sys{} collects time-stamped sensor measurements including  accelerometer, gyroscope, and magnetometer, and the received
 signal strength values from the available WiFi access points and Bluetooth. The collected  raw sensor measurements are preprocessed to reduce the effect of noise by applying a low-pass filter, and phone orientation changes by transforming the sensor readings from the mobile coordinate system to the world coordinate system leveraging the inertial sensors \cite{mohssen2014s}. 
 \setlength{\belowcaptionskip}{-9pt}
\begin{figure}[!t]
\centering
\includegraphics[width=0.51\textwidth]{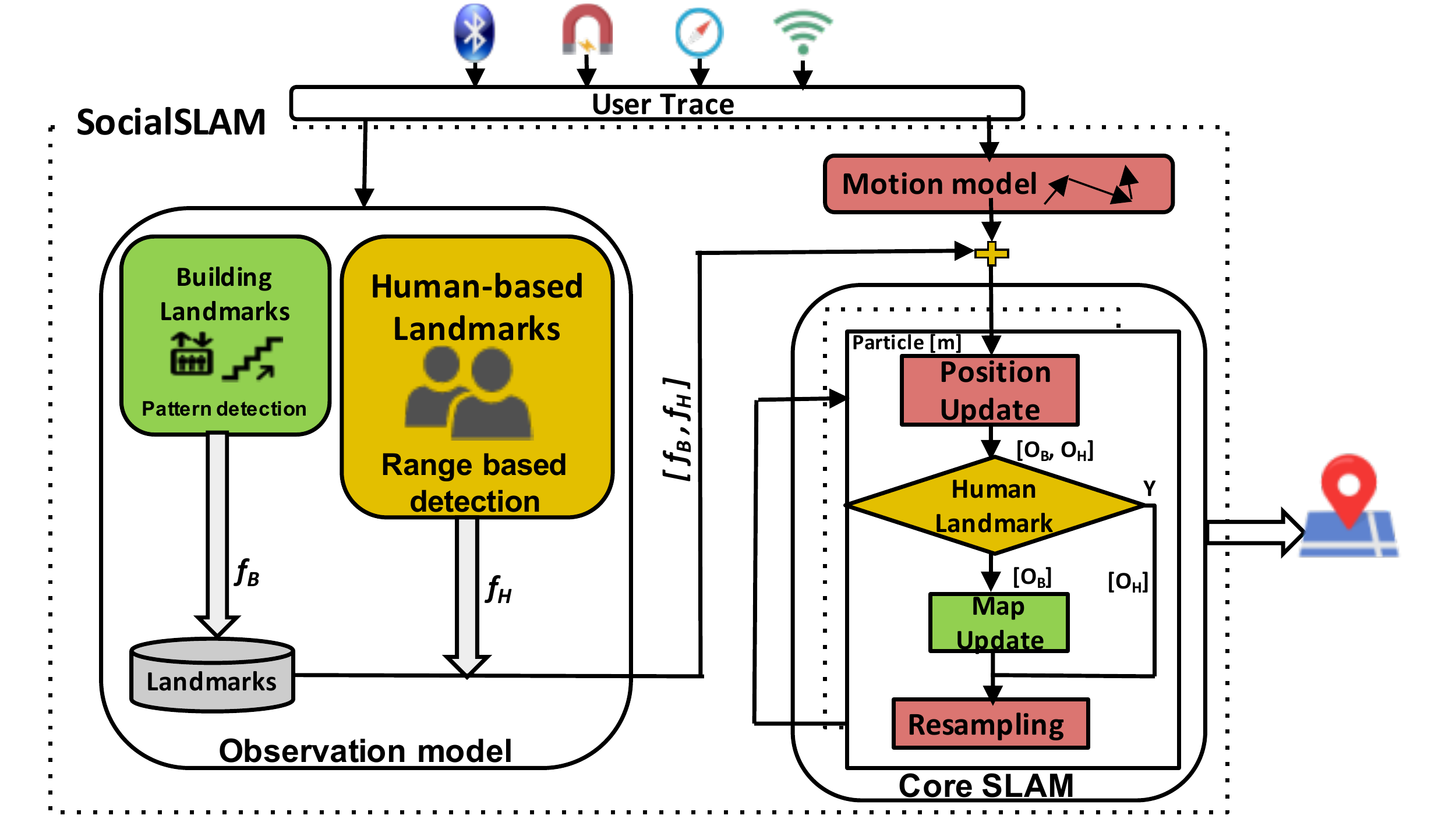}
\caption{\sys{} architecture.}
\label{overview}
\end{figure}
\subsection{Dead reckoning- Motion Model}
\sys{} relies on dead-reckoning to track the user's location starting from a reference point (e.g., the last GPS fix at the building entrance \cite{alzantot2012crowdinside}).  Dead-reckoning  leverages inertial sensors (accelerometer, magnetometer, and gyroscope) where  the current user's location is estimated from  
 the previous location, distance traveled by the user, and the direction of  her motion  during the traveling distance. The user direction of motion  can be estimated from the magnetometer and/or the gyroscope ~\cite{wang2012no},
 while the displacement can be obtained from the accelerometer~\cite{alzantot2012uptime}.
\subsection{Anchors Detection- Observation Model}
The dead-reckoning accumulated location error is unbounded \cite{wang2012no}. Hence, it is infeasible for indoor tracking. To mitigate this error, we propose a anchor-based displacement error resetting technique. In particular, we leverage ample and unique  points in the indoor environment to reset the accumulated error when the user hits one of them. These anchor points (anchors) can be estimated in a crowd-sourced manner, removing the need for any labor-intensive calibration. Traditional anchors are either physical structure/machines or organic anchors that exist in  \textbf{\textit{buildings}}
 (e.g., escalators, elevators, stairs, turns, electronic machines, etc.) \cite{abdelnasser2016semanticslam,wang2012no}. \sys{} \textit{introduces a novel class of} \textbf{\textit{human-based}} anchors, where  encountering other system users is used as dynamic mobile anchors, increasing the anchor density and making them more ubiquitously available. 
 The details of extracting the two anchor classes and handling the associated challenges are discussed in the next section.
\subsection{SLAM Algorithm}
\sys{} extends the FastSLAM algorithm to incorporate its motion and observation models.  Specifically, when  \sys{} detects a user step, it triggers a position update where each particle  updates  its location belief  as detailed in Section \ref{sec:social}.  
When the system detects a human anchor (i.e.,  encountering  another user), the particle location is refined by applying an extended Kalman filter on the relative users' locations and the covariance matrix representing the uncertainty in relative location estimates. The EKF helps to filter out  the noise in users relative location estimation. \textbf{Two approaches} that trade-off the accuracy and efficiency are proposed in \sys{} to fuse the two users' locations. 
On the other hand, if the particle  detects a traditional building anchor, both the  particle location and the map are updated to account for the detection of these anchors.
Finally, in both cases, the particles importance weights are estimated to reflect the confidence of the location associated with each particle. The
weights of  particles are initialized equally and updated at each step based on the likelihood of the
detected anchors (\textit{user} and \textit{building} anchors). Finally, a resampling step is performed using the current
weights in order to refine the particles and drop those that significantly deviated from the actual user location. In the next section, we will explain the details of the SLAM algorithm incorporated in \sys{}.
\section{The DynamicSLAM System}
\label{sec:social}
This section elaborates on the main components of \sys{}: the anchor detection and the core SLAM algorithm as well as provides a proof of the system convergence. Finally, we discuss the possible ways to deploy the system. 
\subsection{Anchor Detection}
In this section, we first briefly discuss the detection of the traditional  building anchors. We then elaborate on the details of the detection algorithm for the \textbf{\textit{novel human-based anchors}}. 
 \subsubsection{Building anchors}
Buildings have many unique points which present identifiable signatures on one or more cell-phone sensors. These unique patterns can be leveraged to automatically recognize these anchors. Once a user sees a anchor, \sys{} detects its signature and leverages it to reset the dead-reckoning error.
The list of  physical  anchors starts by the building's door entrance, detected by the loss of the GPS signal, which is used as
a reference point to seed the dead-reckoning process \cite{alzantot2012crowdinside}. Other  indoor anchors include escalators, elevators, stairs, turns, and doors which can be detected using a finite state machine by employing the approach in \cite{wang2012no}. 
Also, there may exist some anchors that are  specific to some buildings. For instance, railway stations and airports have
many specific machines (e.g., ticketing  machines, lockers, electronic gates, etc) that can be detected  using the different phone sensors \cite{elhamshary2016transitlabel}. 

Finally, there are many organic anchors that may have a unique signature on the phone sensors without  having a physical meaning to humans, e.g., areas with magnetic distortion caused by  nearby electromagnetic waves. These points (organic anchors) can be detected using unsupervised learning techniques and can also be leveraged as anchors \cite{abdelnasser2016semanticslam}.

Note that since both anchor types (physical and organic) have a unique signature in the building, there is no need to differentiate between the different types of building anchors during processing. This differentiation is only relevant during the anchor extraction process.

\sys{} learns anchors locations through crowd-sourcing, which converges to the true anchor location over time \cite{wang2012no,abdelnasser2016semanticslam}. In addition, in some buildings, a map with the location of the building anchors, such as elevators and stairs, may be available. This can further speed the convergence process. Nonetheless, \sys{} models the uncertainty of the anchor locations in its operation. 
\subsubsection{Human-based anchors}
\label{sec:humanLM}
The density of anchors differs from one building to another.  If the  anchors are not dense in a given building, there are fewer opportunities for curbing the dead-reckoning accumulated error, affecting the overall localization accuracy. To alleviate this problem, \sys{} introduces the concept of \textit{\textbf{human-based mobile anchors}}, where other nearby users in the environment are used to reset the localization error. 
 The main intuition is that,  when two users encounter each other, they can calculate their relative distances, and consequently, both users can refine their internal beliefs;  reducing each other  uncertainty in their estimated locations.  \sys{} presents \textit{\textbf{two different measurement models}} to detect the users' encounter and estimate their  relative distances based on the available signals and leading to different ubiquity and accuracy. These are the Bluetooth and the WiFi location models. 
 \begin{enumerate}  
  \itemsep0em
 \item \textbf{The Bluetooth Model:} A user overhearing Bluetooth signals emitted from nearby users can directly infer their encounter. The
 Received Signal Strength  (RSS) can be used to estimate the distance between the nearby users. 
We have conducted experiments to capture the
relation between the Bluetooth RSS and the relative
distance between users. We also show the quality of fitting using different fitting functions, including those based on the traditional known propagation models \cite{bahl2000radar,feldmann2003indoor,jung2013distance}.

Figure \ref{bModel} shows the results. The figure confirms that the quadratic fitting function leads to the best approximation. This is consistent with previous results reported  in literature \cite{feldmann2003indoor,jung2013distance}.

The
user encounters detection can be performed locally on
the device and forwarded to the server for processing.
\item \textbf{The WiFi Model:}
 This model is based on the WiFi signals. The idea is that users in the vicinity of each other should hear similar APs  with close RSS values. Therefore, to detect the encounter of two users, we check the similarity of their collected WiFi RSS measurement vectors. If the similarity is greater than a certain threshold, we detect the users' encounter. 

\begin{figure}[!t]
        \centering
        \begin{subfigure}[b]{0.24\textwidth}
                \includegraphics[width=1\textwidth]{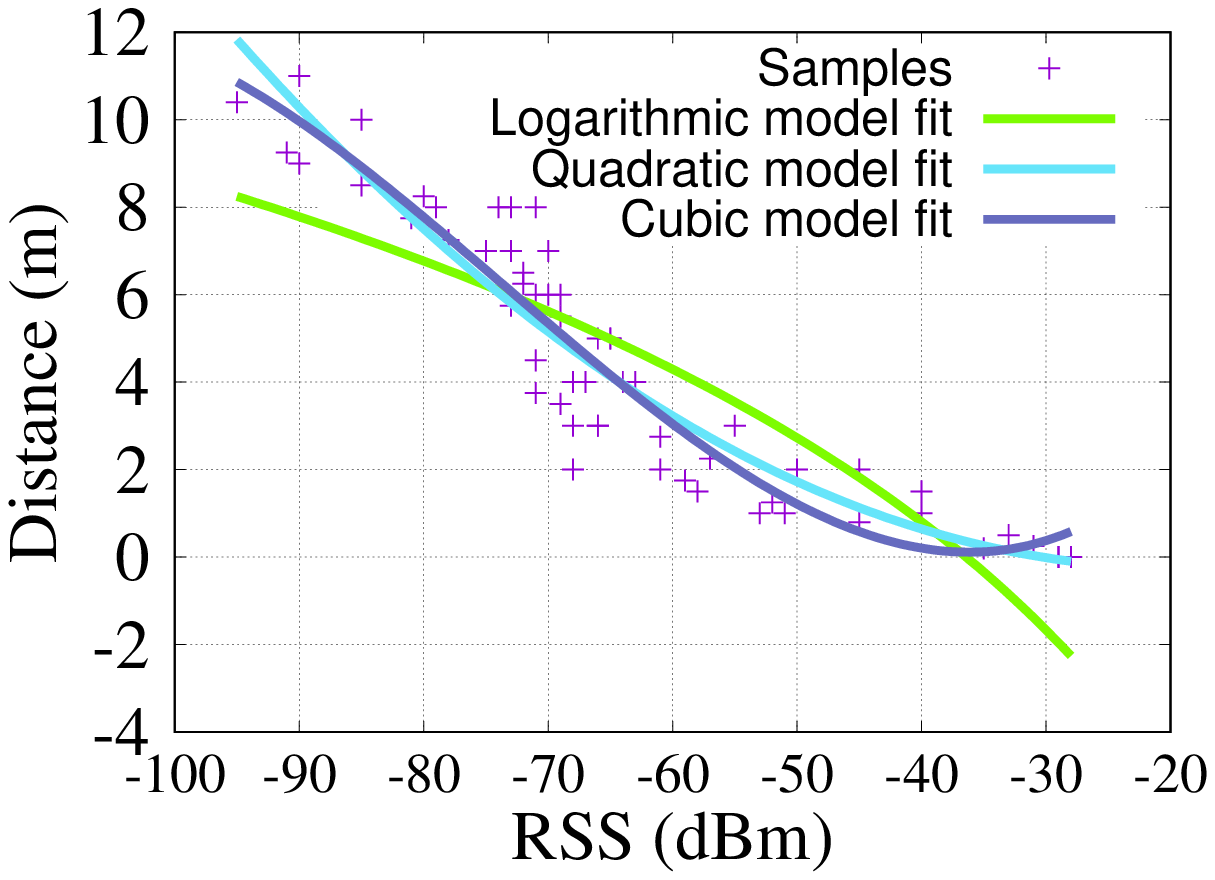}
           \caption{Fitting to raw data.}
           \label{bModel1}
        \end{subfigure}%
        \begin{subfigure}[b]{0.24\textwidth}
                \includegraphics[width=\textwidth]{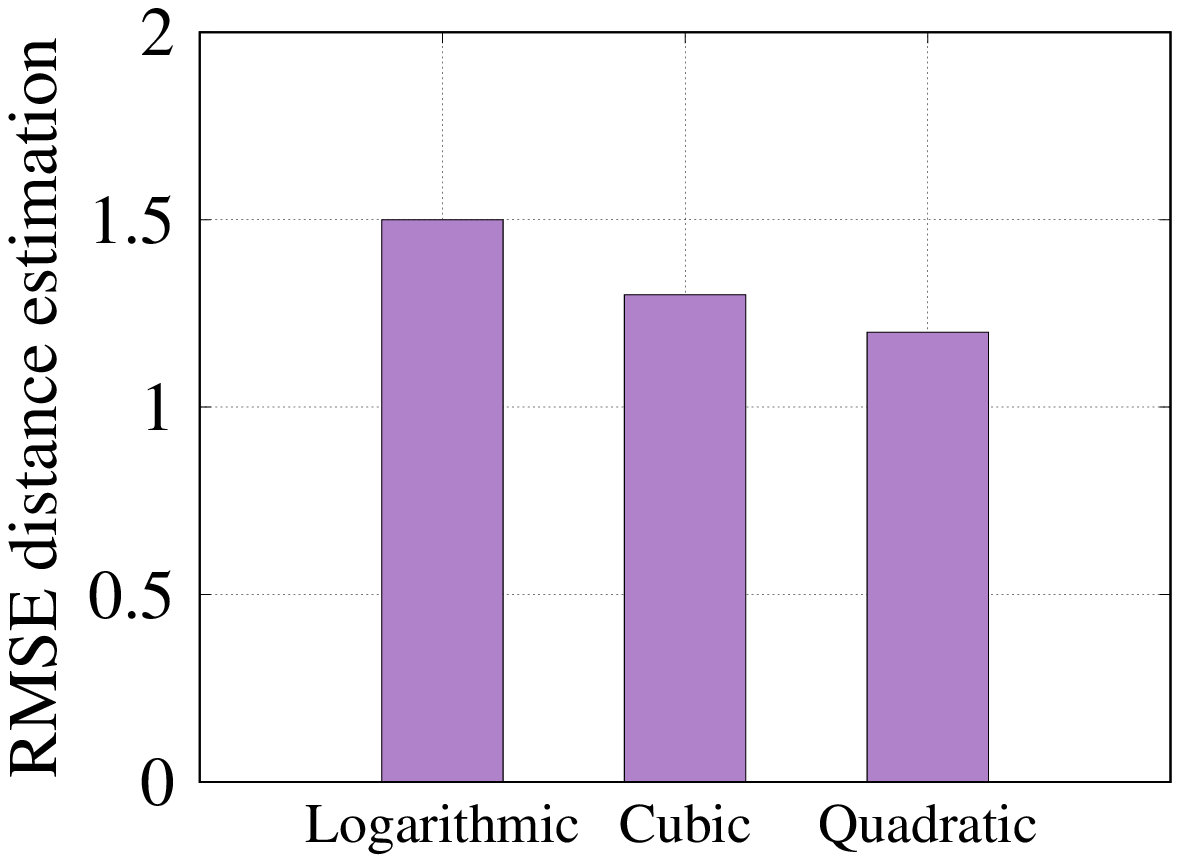}
              \caption{RMSE distance estimation for different Bluetooth models.}
            \label{bModel2}
        \end{subfigure}
        \caption{Comparison between different Bluetooth models.}
             \label{bModel}
\end{figure}

\begin{figure}[!t]
        \centering
        \begin{subfigure}[b]{0.24\textwidth}
                \includegraphics[width=1\textwidth]{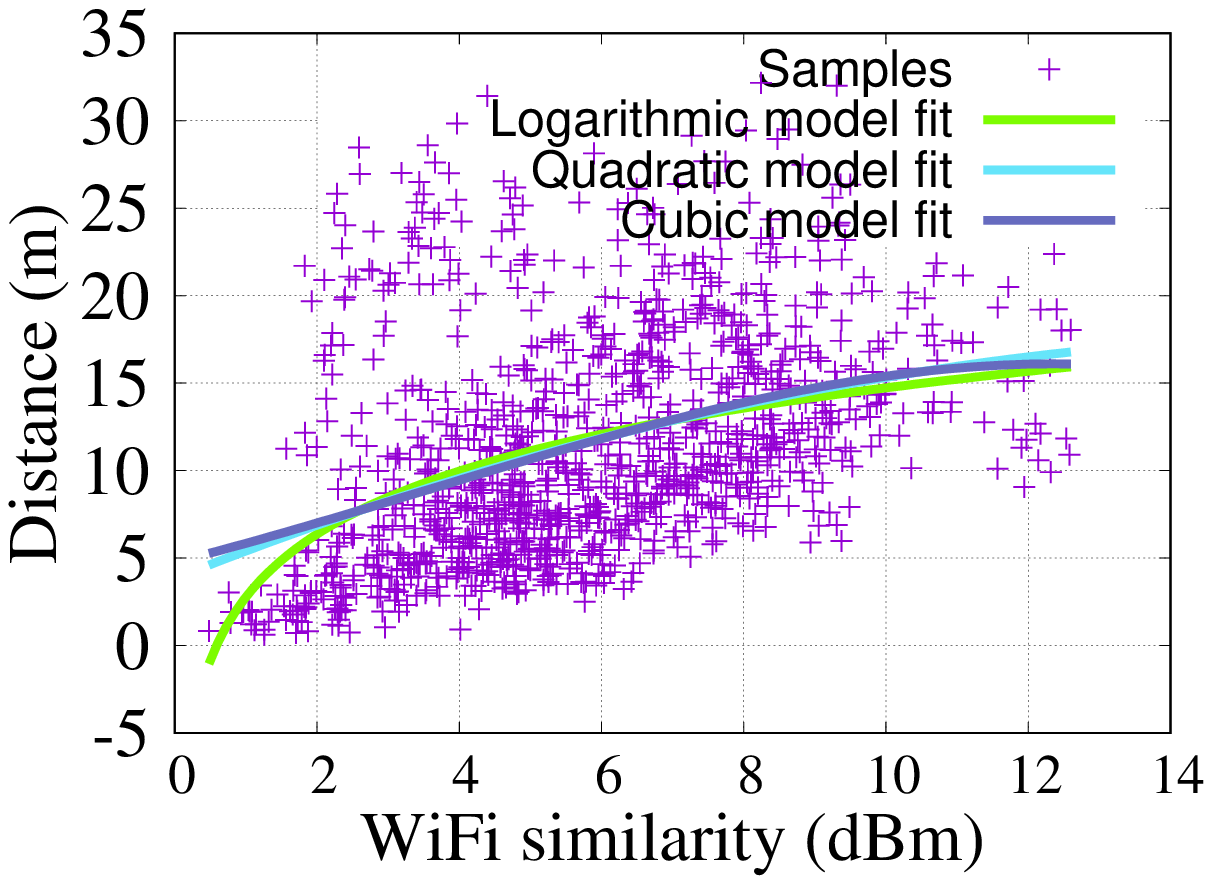}
           \caption{Fitting to raw data.}
           \label{wModel1}
        \end{subfigure}%
        \begin{subfigure}[b]{0.24\textwidth}
                \includegraphics[width=\textwidth]{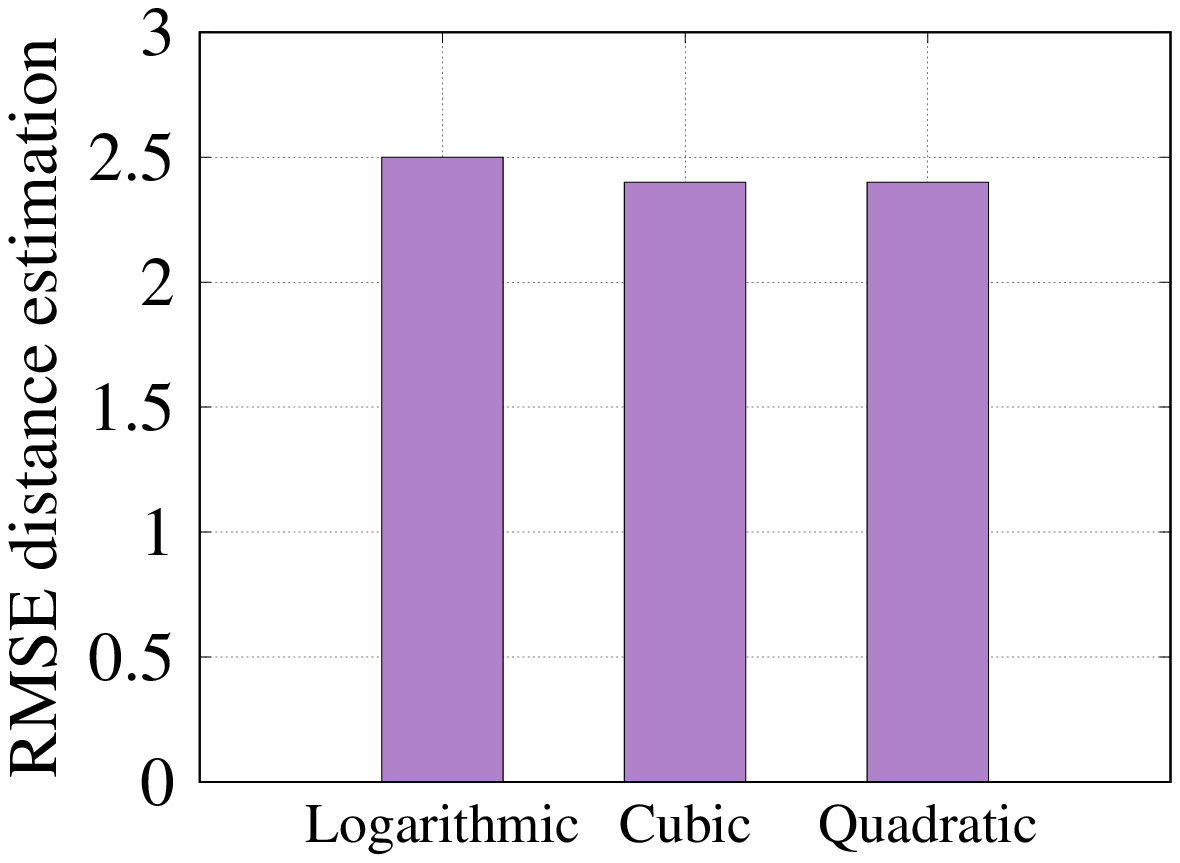}
              \caption{RMSE distance estimation for different WiFi models.}
            \label{wModel2}
        \end{subfigure}
        \caption{Comparison between different WiFi models.}
             \label{wModel}
\end{figure}

 To determine this threshold, we conducted an experiment to build  a spatial physical-distance to WiFi-space distance model. 
Our WiFi-similarity measure takes into account the wireless channel noise that may lead to different number of APs heard in the same location at different times. Specifically, the WiFi similarity measure is calculated as:
 \begin{equation}
 \textrm{Sim}_{\textrm{WiFi}} = \frac{||U_1-U_2||}{n}    
 \end{equation}
 where $n$ is the number of \textit{\textbf{commonly} heard} access points by both users and $U_i$ is the RSS vector of the common APs heard by user $i$. Note that the normalization by the number of APs makes the similarity measure independent of the number of heard APs. 
 Moreover, one can restrict the RSS vector to the strongest $N$ APs to reduce the noise effect. This is quantified in Section~\ref{sec:eval}.

Figure~\ref{wModel1}  shows the relation between the spatial space
Euclidean distance and the WiFi-space similarity measure. Figure~\ref{wModel2} further compares the accuracy of the estimated relative users’ distance using different fit models. 
The figure shows that different models give comparable performance. We therefore use the logarithmic model as it is commonly used in the propagation literature \cite{bahl2000radar}.
The figure also shows that
the WiFi samples are noisier than the Bluetooth samples. This is due to the higher range of WiFi signals
compared to Bluetooth, which makes them more susceptible to noise.

 Finally, since WiFi chips cannot scan for other users' WiFi signal during their normal operation\footnote{WiFi scans are performed between the WiFi chip on the client and the APs installed in the environment to determine the best AP to associate to. To scan for other nearby users' WiFi signals, the device needs to be in monitoring mode, which disables the normal transmission operation.}, WiFi-based users' encounters are performed at the \sys{} server, where the WiFi scans of nearby APs are collected from each user and a module on the server calculates the similarity and performs the encounter detection. 
\end{enumerate}
 \setlength{\belowcaptionskip}{2pt}
\begin{table}
\centering 
\caption{definitions and notations}\label{t1}
\scalebox{0.9}{
\begin{tabularx}{\linewidth}{| c | X |} \hline
	\centering 
	Notation  & Definition\\ \hline \hline
	
	$\hat{u}_t$      & Control data at time $t$\\ \hline   
	
	$\hat{l}_t$      &   Estimated displacement at time $t$ from the sensors\\ \hline
	
	$\hat{L}_t^{[m]}$      &  Estimated step length for particle $[m]$ at time $t$\\ \hline
	
	$\hat{\phi}_t$      &   Estimated heading change  at time $t$ \\ \hline
	
	$\Phi_t^{[m]}$      &  Sampled heading change for the particle $[m]$ at time $t$\\ \hline
	
	$\hat{s}_{t}^{[m]}$  & Predicated location  at time $t$ for the particle $[m]$\\ \hline
	${s}_{t}$  &   User location  at time $t$\\ \hline
	${s}_{t}^{a, [m]}$  &   User $a$'s location as seen by particle $[m]$ at time $t$\\ \hline
	$s_{n,t}^{[m]}$  &   User location as seen by particle  $[m]$ after refining using anchor $n$ at time $t$\\ \hline
	${S}^{t}$  &  entire location  history ($s^t=s_1,s_2,....)$\\ \hline

	$\mu_{s_{t}}^{ a, [m]}$  &  Mean of estimated location $s_t$ of particle $[m]$ of user $a$ at time $t$\\ \hline
	
	$\Sigma_{s_{t}}^{ a, [m]}$  &  Covariance matrix of estimated location $s_t$ of particle $[m]$ of user $a$ at time $t$\\ \hline
	
	$\mu_{s_{n,t}}^{[m]}$  &  Mean of estimated location $s_t$ of particle $[m]$ given anchor $n$ at time $t$\\ \hline
	
	$\Sigma_{s_{n,t}}^{[m]}$  &  Covariance matrix of estimated location $s_t$ of particle $[m]$ given anchor $n$ at time $t$\\ \hline

	$z_{n,t}^a$  &   \textbf{Estimated} measurement of  anchor $n$ position at time $t$ by  user $a$ \\ \hline
	$\hat{z}_{n,t}^{a,[m]}$  &    \textbf{Predicted} measurement  of particle $[m]$ for user $a$ of observed  anchor $n$ at time $t$ \\ \hline 
	
	$Q_{n,t}^{ [m]}$  &   Covariance matrix of observing anchor $n$ at time $t$ by particle $[m]$\\ \hline
	
	
	$P_{t}$  & Covariance matrix of control data at time $t$ \\ \hline
	
	$R_t$  &   Measurement covariance at time $t$ \\ \hline
	
	$K_{t} ^{[m]}$  &  Kalman gain for particle $[m]$ at time $t$\\ \hline

	$\theta$& anchor map\\ \hline
	$\theta_{n_t}$& Location of anchor $n$ at time $t$\\ \hline
	
	$N_{t}^{[m]}$  &  Number of anchors for particle $[m]$ at time $t$\\ \hline

	$\mu_{n,t} ^ {[m]}$ &  Mean of anchor $n$ location at time $t$ for particle $[m]$ \\ \hline
	
	$\Sigma_{n,t}^{[m]}$ &  Covariance matrix of anchor $n$ at time $t$ for particle $[m]$\\ \hline
	
	$ w_t^{[m]}$  &  Weight of particle $[m]$ at time $t$ \\ \hline
	
	$f_t$  &  anchor type (e.g. elevator) \\ \hline
	
	$\hat{f}_t$  &  Detected anchor type (e.g. elevator pattern) \\ \hline
	
	$ p_0$&  Likelihood of observing a anchor for the first time \\ \hline
	
	$ p_n $ &  Likelihood correspondence of anchor $n$ with the observed pattern $\hat{f}$ at time $t$ \\ \hline
	
	$ \eta $ & Normalization factor  \\ \hline
	
	$  J_{n,t}$&  The Jacobian   of  measurement model  for anchor $n$ at time $t$ \\ \hline
	
	$ \beta $ & A threshold on the number of iterations for the iterative sampling technique \\ \hline

\end{tabularx}
}
\label{tab:not}
\end{table}
\subsection{Core SLAM Algorithm}
\sys{} is a variant of the FastSLAM \cite{montemerlo2002fastslam}
framework 
 optimized to operate with  anchors that can be detected based on the limited capabilities of phone-embedded sensors, as compared to robot sensors. 
\sys{} uses the particle filter to enhance user's position. It relies on dead-reckoning to track the user’s location and \textbf{predicts} the distance (i.e. offset) between the user location and anchor location (Motion model / Prediction model). On the other hand, it relies on the Bluetooth/WiFi models to get the \textbf{estimated} distance (Observation model). Finally, \sys{} combines both models using the particle filter equations.  
 
 \sys{} core algorithm consists of three steps: location sampling (motion update),
map update (observation update), and re-sampling.  Without loss of generality, we assume that only a single anchor is observed at any time instance $t$ to simplify the problem complexity. Note that, since DynamicSLAM needs around 150ms to get the estimated location while the typical human speed is 1.4m/s; the anchor can be assumed stationary during the estimation process. Table \ref{tab:not} summarizes the notations used in the paper.
\subsubsection{Motion Update}
It is invoked when a step is  detected based on the phone inertial sensors. At this time instance, the control signal $\hat {u}_t=<\hat{l}_t, \hat{\phi}_t>$, where $\hat {l}_t$ and $\hat {\phi}_t$ are the  estimated  step length and heading from the phone inertial sensors, updates the user's location from $s_{t-1}$ to $s_t$.  More formally,  given the control signal $\hat {u}_t$ at time $t$, we draw each particle $[m]$ according to:
 \begin{equation}
\label{eqsamp}
\hat{s}_t^{[m]} \sim P(s_t^{[m]}|s_{t-1}^{[m]},\hat{u}_t)
\end{equation}
Assuming Gaussian-distributed error for both $\hat{l}_t$
and $\hat{\phi}_t$, the sampled displacement and heading for particle $[m]$ at time $t$ are calculated as:
\begin{equation}
L_t^{[m]} \sim N(\hat{l}_t,\sigma_{l})\  ,\  
\phi_t^{[m]} \sim N(\hat{\phi}_t,\sigma_{\phi})
\end{equation}
where  
 $\sigma_{l}$ and $\sigma_{\phi}$ are the variance of the displacement and heading estimation errors respectively.
Therefore, the  new  sampled  location for a given user  in Eq.~\ref{eqsamp} can be rewritten as:
\begin{equation}\label{motion0}
s_t^{[m].\phi} = s_{t-1}^{[m].\phi} + \phi_t^{[m]}
\end{equation}
\begin{equation}\label{motion1}
s_t^{ [m].x} = s_{t-1}^{[m].x} + L_t^{[m]}. \cos(s_t^{[m].\phi})   
\end{equation}
\begin{equation}\label{motion2}
s_t^{[m].y} = s_{t-1}^{[m].y} + L_t^{[m]}. \sin(s_t^{[m].\phi})   
\end{equation}
where $s_t^{[m].\phi}, s_t^{ [m].x} , s_t^{[m].y} $ are the new user heading, x, and y coordinates of the location respectively after incorporating the  control input $\hat{u}_t$. 
\subsubsection{Observation Update}
At any time instance, the user may observe a human-based ($O_{LH}$) or building ($O_{LB}$) anchor. In both cases, the current user's location belief is updated. In addition, for the case of the building anchor, the feature map is also updated to refine the observed anchor location. The balance of this subsection describes the details of both cases.

\textbf{Observing a Human-based anchor}
When user  \textit{Alice} encounters another user  \textit{Bob}, her location should be updated based on their relative distances $z_t$ and this distance covariance matrix, i.e., the confidence in the estimated location. Note that, since each user is considered a anchor for the other user, \textbf{no map update} is invoked in this case. To  update \textit{Alice}'s  location  based on Bob's encounter, we propose two different techniques that trade-off the accuracy and efficiency: One-shot sampling and iterative sampling.

\begin{enumerate}
\item{One-shot Sampling:} \\
To account for the relative distance  measurement error between the two users, we incorporate the EKF to update the user's location  belief based on the observation. The EKF captures the deviation of the predicted measurement, $\hat{z}_{b,t}^{a, [m]}$, (based on the estimated users' locations) from the estimated measurement, $z_{b,t}^{a}$,  (based on the two observation models described in Section~\ref{sec:humanLM}), taking into account the confidence in the different quantities

Specifically, \textit{Alice}'s  \textbf{predicted} distance measurement, i.e., the location offset between \textit{Alice}'s (a) and \textit{Bob}'s  (b) locations is calculated as: 
   \begin{align}\label{observationz}
\hat{z}_{b,t}^{a,[m]}  & = ||\hat{s}_{t}^{a,[m].(x,y)} - s_{t}^{b,[m].(x,y)}||
\end{align}

 The mean and covariance matrix of \textit{Alice}'s location  distribution  after incorporating \textit{Bob}'s encounter is:
\begin{equation}\label{sam-prob-mu}
\mu_{s_{b,t}}^{a,[m]} =\hat{s}_{t}^{a,[m]} + K_t^{[m]} (z_{b,t}^{a}-\hat{z}_{b,t}^{a,[m]}) 
\end{equation}
\begin{equation}\label{sam-prob-cov}
\Sigma_{s_b,t} ^{a, [m]}  = [\ J_{s_b,b}^T\  (Q_{b,t}^{ [m]})^{-1} J_{s_b,b}^T \ +\  P_{t}^{-1}]^{-1}
\end{equation}	
Where $J_{s_b,b}^T$  is the transpose of the Jacobian of the measurement model with respect to Bob's location $\mu_b= s_{t}^{b,[m].(x,y)}$, $P_t$ is the covariance matrix of the control data at time $t$, $Q_{b,t}^{[m]}$ is the human-based anchor observation covariance matrix given by:
\begin{equation}
Q_{b,t}^{[m]} = J_{b,t} \Sigma_{b,t-1}^{[m]} J_{b,t}^T+ R_t
\end{equation}
and $R_t$ is the measurement covariance matrix. Finally, $K_t^{[m]}$ is the Kalman gain given by:
\begin{equation}
K_t^{[m]}=\Sigma_{s_{b,t}} ^{a, [m]}\  (J_{s_b,b}^{a})^T\  (Q_{b,t}^{[m]})^{-1}\ 
\end{equation}
Finally, the new location of a particle $[m]$ of \textit{Alice} after incorporating Bob's encounter is sampled from the following distribution: 
\begin{equation}
         s_{b,t}^{a,[m]} \sim N(\mu_{s_{b,t}}^{a,[m]},\Sigma_{s_{b,t}}^{ a, [m]})
 \end{equation}
Intuitively, when a user encounters another user, her location is updated based on her previous state, the confidence of the two users' locations, and the innovation (the difference between the estimated and predicted relative distance). Algorithm~\ref{alg:refine} summarizes the process of refining the particle position.
\begin{algorithm}[!t]
	\caption{Refine\_Particle\_Position ($\hat{s}^{[m]}_{t},[m],n)$} 
	\label{alg:refine} \small
        \begin{algorithmic}[1]
		\Require
	    \Statex $\hat{s}^{[m]}_{t}$: Predicated location at time $t$ for the particle $[m]$.
	    \Statex $[m]$: Particle index.
	    \Statex $n$: Anchor index.
	    \newline
	    	\State $\hat{z}_{n,t}^{[m]} = ||\mu_{n,t-1}^{[m]}- \hat s_{t}^{[m]} ||$ \Comment{predict measurement}
		\State $K_t^{[m]}= \Sigma_{s_{n,t}}^{[m]}\ J_{s,n}^T (Q_{{n},t}^{[m]})^{-1}$   \Comment{Kalman gain}
		
		\State $\Sigma_{s_n,t}^{[m]}  = [\ J^T_{s,n}\  Q_{n,t}^{ [m]-1} J_{s,n} \ +\  P_{t}^{-1} ]^{-1}$ \State\Comment{covariance of proposal distribution}

		\State $\mu_{s_n,t}^{[m]}  =\hat{s}^{[m]}_{t} +K_t^{[m]}\ (z_{n,t}-\hat{z}_{n,t}^{[m]})$ \State\Comment{mean of proposal distribution}

		\State $s_{n,t}^{[m]} \sim N(\mu_{s_n,t}^{ [m]},\Sigma_{s_n,t}^{ [m]})$ \Comment{sample pose}
	    
	\end{algorithmic}
\end{algorithm}
\begin{algorithm}[!t]
	\caption{DynamicSLAM ($z_t,u_t,S_{t-1})$} 
	\label{alg:socialSLAM} \small
	\begin{algorithmic}[1]
	    \Require
	    \Statex $S_{t-1}$: System state at time $t-1$. 
	    \Statex $u_t$: Control data at time $t$.
	    \Statex $z_t$: Measurement observation at time $t$.
	    \newline
	    \For{m = 1 to M} \Comment{particle m}
		\State Retrieve $\big \langle \langle s_{t-1}^{[m]},w_{t-1}^{[m]}\rangle,
		\langle \mu_{1,t-1}^{[m]},\Sigma_{1,t-1}^{[m]},f_{1,t-1}^{[m]}\rangle,..,\newline
		\langle \mu_{N,t-1}^{[m]},\Sigma_{N,t-1}^{[m]},f_{N,t-1}^{[m]}\rangle \big \rangle$ from $S_{t-1}$  \State \Comment{get particle's old state}
		
		\State  $\hat s_t^{[m]} \sim P(s_t|s_{t-1}^{[m]},\hat{u}_t)$ \Comment{motion model- predict  pose}

		\If {Human anchor $h$ is observed} \State \Comment{observation model}

		\State \textit{Refine\_Particle\_Position}($\hat{s}_{t}^{[m]},[m],h$) \State \Comment{refine user's pose}

		\Else \Comment{building anchor}

		\For{n = 1 to $N_{t-1}^{[m]}$} \Comment{calculate sampling distribution}
		\State \textit{Refine\_Particle\_Position}($\hat{s}_{t}^{[m]},[m],n$) 
		\State \Comment{refine user's pose}
		\State $p_n =|2\pi Q_{n,t}|^{ -\frac{1}{2}}exp\{-\frac{1}{2} (\hat{z}_{n,t}^T .Q_{n,t}^{-1} \hat{z}_{n,t}\}p(f_t|\hat{f}_t)$  \State \Comment{association likelihood} 
		\EndFor
		
		\State $p_{1+N_{t-1}^{[m]}} = \eta p_0 p(f_t|\hat{f}_t) $ \Comment{likelihood of new anchor}
		\State $\hat{n} = argmax_{\{n \in 1,..,1+N_{t-1}^{[m]}\}} p_n$ \Comment{maximum likelihood}
		\State $N_{t}^{[m]} = \max\{N_{t-1}^{[m]} , \hat{n}\}$ \Comment{new number of anchors}
		\State \textit{Map\_Update}($N_{t-1}^{[m]}, \hat{n}, [m]$) \Comment{map update}
		\State $w_t^{[m]} = w_{t-1}^{[m]} . p_{\hat{n}}^{[m]}$ \Comment{importance weight}
		\EndIf
		\EndFor
		\State \textbf{loop} M times  
		\State \ \ \  \ Draw with replacement a random index $m$ with probability $\propto\  w_t^{[m]}$  \Comment{resampling}
		\State \textbf{end loop}
		\State add $\big \langle \langle s_{t}^{[m]},w_{t}^{[m]}\rangle,
		\langle \mu_{1,t}^{[m]},\Sigma_{1,t}^{[m]},f_{1,t}^{[m]}\rangle,....,\langle\mu_{N,t}^{[m]},\Sigma_{N,t}^{[m]},f_{N,t}^{[m]}\rangle \big \rangle$ to $S_{t}$
	\end{algorithmic}
\end{algorithm}   
\begin{algorithm}[!t]
	\caption{Map\_Update ($N_{t-1}^{[m]}, \hat{n}, [m]$)} 
	\label{alg:map_update} \small
	\begin{algorithmic}[1]
		\Require
	    \Statex $N_{t-1}^{[m]}$: Number of anchors seen by particle $[m]$ at time $t$-$1$. 
	    \Statex $\hat{n}$: anchor index with the highest probability.
	    \Statex $[m]$: Particle index.
	    \newline
		\If {$\hat{n} = 1 + N_{t-1}^{[m]} $} \Comment{New anchor detected}
		\State $\mu_{\hat{n},t}^{ [m]} = s_t^{[m]}$ \Comment{initialize its mean}
		\State $\Sigma_{\hat{n},t}^{[m]} = R_t$ \Comment{initialize its covariance}
		\Else 
		\State $K_t^{[m]}= \Sigma_{\hat{n},t-1}^{[m]}  (Q_{\hat{n}_t,t}^{[m]})^{-1}$ \Comment{Kalman gain}
		\State $\mu_{\hat{n},t}^{[m]} = \mu_{\hat{n}_t,t-1}^{[m]} - K_t^{[m]}\  \hat{z}_{\hat{n},t}^{ [m]} $ \Comment{update mean}
		\State $\Sigma_{\hat{n},t}^{[m]} = (I-K_t^{[m]}) \Sigma_{\hat{n},t-1}^{[m]} $ \Comment{update covariance}
		\EndIf
		
	\end{algorithmic}
\end{algorithm}
\item{Iterative Sampling:}\\
In this technique, the one-shot sampling described in the previous section is repeated on the same pair of users until the particle position converges. To determine convergence, the algorithm stops when the difference between two successive location updates for particle $[m]$ is less than a certain threshold $\beta$.

The iterative sampling approach has the advantage of increased accuracy at the expense of increased computations (as we quantify in Section~\ref{sec:eval}).
\end{enumerate}
\textbf{Observing a Building Anchor: }
Similar to the human-based anchors, encountering a building-based anchor, e.g., an elevator, updates the user's  position using the EKF. In addition, the system needs to (I) resolve the uncertainty in the type of detected anchors, i.e., stairs vs elevators vs other types, due to the inherent noise in the phone sensors and (II) Update the building anchor location.

(I) Data Association (Detecting building anchor type)
To account for this, we compute the conditional probability of actually observing each anchor type $(f)$ when the detected
pattern is of type $\hat f$. This uncertainty is a function of the sensors noise and is independent from the building. It is typically described by a confusion matrix, where each entry indicates the probability $p(f|\hat{f})$ of the anchor $(f)$ given pattern $\hat{f}$.
Given this confusion matrix and the anchor location uncertainty $Q_n$ for each anchor $n$ in the particle's map,
the likelihood of correspondence with the observed pattern  is calculated as \cite{abdelnasser2016semanticslam}
\begin{equation}
p_n^{[m]} = \eta |2\pi Q_{n,t}^{[m]}|^{-\frac{1}{2}} exp{\{-\frac{1}{2} (\hat{z}_{n,t}^{[m]})^{T} Q_{n,t}^{[m]-1} \hat{z}_{n,t}^{[m]}\}}.p(f_t|\hat{f}_t)
\end{equation}
Where $\eta$  is a normalization factor. This
likelihood takes into account the distance
between the current user's location and the building anchor, the
location uncertainty of both, and the confusion between the
observed and actual anchor types ($f_t$ vs $\hat{f}_t$).

We assume that the probability of observing
a newly unseen anchor given the detected pattern is
calculated as:

\begin{equation}
p_n^{[m]} =  \eta p_0.p(f_t|\hat{f}_t)
\end{equation}
where $p_0$ is a constant determined empirically.
\noindent \textit{(II) Updating anchor Location (\textbf{Map update})}
The anchor with the highest probability
(${\hat n}_t$) is selected as the observed anchor
and its location is updated using the standard EKF
formulas. (Algorithm~\ref{alg:map_update}).
    \subsubsection{Resampling}
 The last step  assigns a weight to each particle
to reflect the confidence of its  location. The particles' weights are initialized uniformly and
updated at each step based on the likelihood of the detected anchors, i.e., the anchor ${\hat{n}_t}$ with the highest
probability, as:
\begin{equation}
w_{t}^{[m]} = w_{t-1}^{[m]}.p_{\hat{n}_t}
\end{equation}
These weights are normalized to add to one. Then, the location and the map of the particle having the
maximum weight are selected as the current estimate. Finally, the algorithm  draws with replacement $M$ new particles to avoid degeneration.
Algorithm~\ref{alg:socialSLAM} summarizes the full \sys{} algorithm.
\subsection{\sys{} Convergence}
\label{sec:conv}
In this section, we argue about the system convergence as well as the effectiveness of human anchors ability to reset the accumulated error. Here convergence refers to reducing the error in the users' locations to zero. 
Assume that user $x$ encounters an anchor $n$ (human or building anchor). We define
the error in user $x$ location $\Psi_t^{x,[m]}$ as follows:
\begin{equation}
     \Psi_t^{x,[m]} = \mu_{s_t}^{x,[m]} - s_t^x
\end{equation}
Note that the ground-truth user location $s_t^x$ is used in the previous equation only to prove \sys{} convergence and is not needed during the actual system operation.
From equations \ref{motion0}-\ref{motion2} and \ref{observationz},
$\hat{s}_t^{x,[m]} = s_{t-1}^{x,[m]} + u_t$,
$\hat{z}_t^{x,[m]} = \mu_{n,t-1}^{[m]} - s_{t-1}^{x,[m]} - u_t$. Also, by setting
$J_n = I$,
$J_s = -I$, and
$Q_{n,t}^{[m]} = R_t +\Sigma_{n,t-1}^{[m]}$, the mean and covariance of Eq.~\ref{sam-prob-mu} and \ref{sam-prob-cov} of the
proposal distribution are defined as a follows
\begin{equation}\label{mu-sample}
\begin{aligned}
\mu_{s_{n,t}}^{x,[m]} =
-\Sigma_{s_n,t}^{x,[m]} ( R_t +\  \Sigma_{n,t-1}^{[m]})^{-1} \cdot\ 
(z_{n,t}^a\ -\ \mu_{n,t-1}^{[m]}\ +\\ s_{t-1}^{x,[m]}\ +\ u_t)\ +\  
s_{t-1}^{x,[m]}\ +\ u_t
\end{aligned}
\end{equation}
\begin{equation}\label{cov-sample}
\Sigma_{s_n,t} ^{x, [m]}  = [( R_t + \Sigma_{n_t,t-1}^{[m]})^{-1}  \ +\  P_{t}^{-1}]^{-1}
\end{equation}
First, we study the human anchor ability to reset the accumulated error.
\begin{lemma}\label{lem1}
If user $a$ with error $\Psi_t^{a,[m]}$ encounters another user $b$ with error $\Psi_t^{b,[m]}$, where $\Psi_t^{a,[m]} < \Psi_t^{b,[m]}$, then $\Psi_t^{b,[m]}$ will shrink in expectation. Conversely, if $\Psi_t^{a,[m]} > \Psi_t^{b,[m]}$, then the later may increase but will not exceed $\Psi_t^{a,[m]}$ in expectation.   
\end{lemma}
\begin{proof}
The expected error of user $a$ at time $t$ is given by:
\begin{equation}
    E[\Psi_t^{a,[m]}] = E[\mu_{s_t}^{a,[m]}- s_t^a] = E[\mu_{s_t}^{a,[m]}] - E[s_t^a] 
\end{equation}
We can get the first term from the sampling distribution (Eq~\ref{mu-sample}) which depends on user location, anchor location, estimated measurements, and the uncertainty in anchor and user locations. The second term can be obtained from the motion model (Eq~\ref{motion1} and \ref{motion2}). Therefore:
\begin{equation}\label{expecteps}
\begin{aligned}
    E[\Psi_t^{a,[m]}] =
    -\Sigma_{s_b,t}^{a,[m]} ( R_t +\  \Sigma_{b,t-1}^{[m]})^{-1}\ \cdot\  
(E[z_{b,t}^a]\\ -\ \mu_{b,t-1}^{[m]}\ +\ s_{t-1}^{a,[m]}\ +\ u_t)\ +\ 
\Psi_{t-1}^{a,[m]}
\end{aligned}
\end{equation}
For linear SLAM, $E[z_{b,t}^a]= s_t^{b,[m]} - E[s_t^{a,[m]}]
= s_t^{b,[m]} - u_t-s_{t-1}^{a}$. Hence,
\begin{equation}\label{expz}
    \begin{aligned}
        E[z_{b,t}^a] - \mu_{b,t-1}^{[m]} + s_{t-1}^{a} + u_t 
=\  s_t^{b,[m]} - u_t-s_{t-1}^{a} \\ - \mu_{b,t-1}^{[m]} + s_{t-1}^{a,[m]} + u_t
=\ \Psi_{t-1}^{a,[m]} - \Psi_{t-1}^{b,[m]}
    \end{aligned}
\end{equation}
By substitution from equations \ref{cov-sample} and \ref{expz} in \ref{expecteps},
\begin{equation}\label{finalexp}
    \begin{aligned}
        E[\Psi_t^{a,[m]}]
    = \ \Psi_{t-1}^{a,[m]}\ +\ 
     [I\ +\ ( R_t +\  \Sigma_{b,t-1}^{[m]}) P_t^{-1}]^{-1}\\ \cdot\ 
(\Psi_{t-1}^{b,[m]} - \Psi_{t-1}^{a,[m]}) 
    \end{aligned}
\end{equation}
As $R_t$, $\Sigma_{b,t-1}^{[m]}$, and $P_t^{-1}$ matrices  are positive semidefinite, the value of $[I\ +\ ( R_t +\  \Sigma_{b,t-1}^{[m]}) P_t^{-1}]^{-1}$ is a contraction matrix. Hence, $E[\Psi_t^{a,[m]}]$ depends on the difference $\Psi_{t-1}^{b,[m]} - \Psi_{t-1}^{a,[m]}$. If $\Psi_{t-1}^{b,[m]} < \Psi_{t-1}^{a,[m]}$, then user $a$ location error will be reduced in expectation. If $\Psi_{t-1}^{b,[m]} > \Psi_{t-1}^{a,[m]}$, then the $E[\Psi_t^{a,[m]}]$ will increase by a certain amount depending on the difference, which makes the upper bound for the expected error, $E[\Psi_t^{a,[m]}]$,  $\Psi_{t-1}^{b,[m]}$. 
\end{proof}
Now, we study the effect of a user that encounters a building anchor with known location.
\begin{lemma}\label{lem3}
If a user $a$ encounters a building anchor with known location $n^*$, then the error in her location will shrink in expectation.
\end{lemma}
\begin{proof}
The building anchor with known location has zero error $\Omega_{t}^{n^*,[m]}$ = 0 and $\Sigma_{n^*,t}^{[m]}= 0$. By substitution in Eq.~\ref{finalexp}
\begin{equation}\label{finalexpland}
   \begin{split}
    E[\Psi_t^{a,[m]}]
    = \ \Psi_{t-1}^{a,[m]}\ -\ 
     [I\ +\ R_t \cdot P_t^{-1}]^{-1}\cdot \Psi_{t-1}^{a,[m]}
    \end{split}
\end{equation}
\end{proof}
The following theorem proves our system convergence.
\begin{theorem}\label{thm1}
\sys{} converges even for $M$=1 particles if there is at least one building anchor whose location is known in advance.
\end{theorem}
\begin{proof}
Assuming that the building contains at least one of building anchor with known location, some of the system users will eventually observe the building anchors, which according to Lemma 2 will lead to reducing their error. Similarly,  when two users encounter each other, the one with the lower accuracy will enhance his estimate based on the other user (from Lemma 1). Therefore, the error of the different users in \sys{} will always reduce, leading to convergance. 
\end{proof}

\subsection{Discussion}
\sys{} can have different deployment modes. In this paper, we assumed that
	the system operations are performed on a centralized server. However, the
	system can be completely implemented on devices that communicate with each others in a peer-to-peer manner. The
	devices in this case will broadcast their own estimated locations, confidence
	in the estimated locations, and RSS heard from other devices. Note that, unlike the traditional collaborative SLAM-based techniques, \sys{} does not need to exchange particles local maps as the users' locations are updated based on their local maps only.

	The peer-to-peer implementation is more private. It can solve the latency and connection problems due to the centralized process. 
	However, making all computations on smart-phones may be
	an overhead due to the limited computational power on smart-phones. In both cases, cryptographic approaches can be used to prevent malicious users from introducing errors in the system, e.g. by broadcasting fake or wrong locations. 
\section{Evaluation}
\label{sec:eval}
In this section, we describe the testbed and the data collection methodology followed by evaluating the performance of the system components under different parameter settings. We end the section by comparing \sys{} performance against other systems and quantifying its energy footprint.
 \subsection{Testbed \& Data Collection Methodology}
We deployed \sys{} in the engineering building in our university campus. The building area is 3000$m^2$ containing offices, labs, meeting rooms as well as corridors.
 The building has 30 building anchors (physical and organic). The data is collected by 23 users using different Android phones (e.g., LG Nexus
5, Samsung Galaxy Note 3, Galaxy 4, Galaxy Tab, among others). The data is then processed in a centralized server. Each user walked around in a random independent path in the building for at least 250m passing by some building anchors (10 on average) as well as encountering other users (7 on average). 
If a user encounters another group of users, the system selects the user with the highest confidence. 
  \setlength{\belowcaptionskip}{-2pt}
\begin{table}[!t]
\centering
  \caption{Default parameter values}
\label{default}
\scalebox{0.7}{
\begin{tabular}{|l|l|l|}\hline
   Parameter & Range & Default value \\ \hline \hline
   Measurement Model & Bluetooth, WiFi  & WiFi \\ \hline
   Location update technique & One-shot sampling, Iterative sampling  & Iterative sampling \\ \hline
   Number of particles & 1 to 100 & 75 \\ \hline
   User inverse density\textbf{$(m^2/$user)} & 3 to 3000 & 30 \\ \hline
   RSS threshold (Bluetooth model) & -70 to -95 & -90 \\ \hline   
   WiFi threshold (WiFi model) & 2 to 14 & 8 \\ \hline
    Number of APs (WiFi model) & 1 to 16 & 5 \\ \hline
   Iterative sampling threshold $\beta$ & 0 to 1 & 0.1 \\ \hline
\end{tabular}
  }
\end{table}
\subsection{Effect of Changing \sys{} Parameters}
In this section, we study the effect of the different parameters on
the system performance including the user encounter detection method, user density in the building, sampling method, and number of particles.
 Table \ref{default} shows the
default parameters values used throughout the evaluation section.
\subsubsection{Users' encounter observation models}
\textbf{Bluetooth model}
Figure \ref{fig:bluetooth} shows the accuracy of the estimated relative users' distance (in terms of Root Mean Square Error) at different cut-off Bluetooth RSSs. The figure shows that the higher the cut-off threshold, i.e. the stronger the Bluetooth signal, the lower the relative distance error. However, requiring a strong Bluetooth signal reduces the users' encounter opportunities due to the lower range. For example, setting the cut-off RSS threshold to -70dBm, the encounter detection range is around 4m (From Figure~\ref{bModel}), and the RMSE of relative distance is around 1m.
\setlength{\belowcaptionskip}{-7pt}
\begin{figure*}[!t]
\minipage{0.23\textwidth}
  \includegraphics[width=\linewidth]{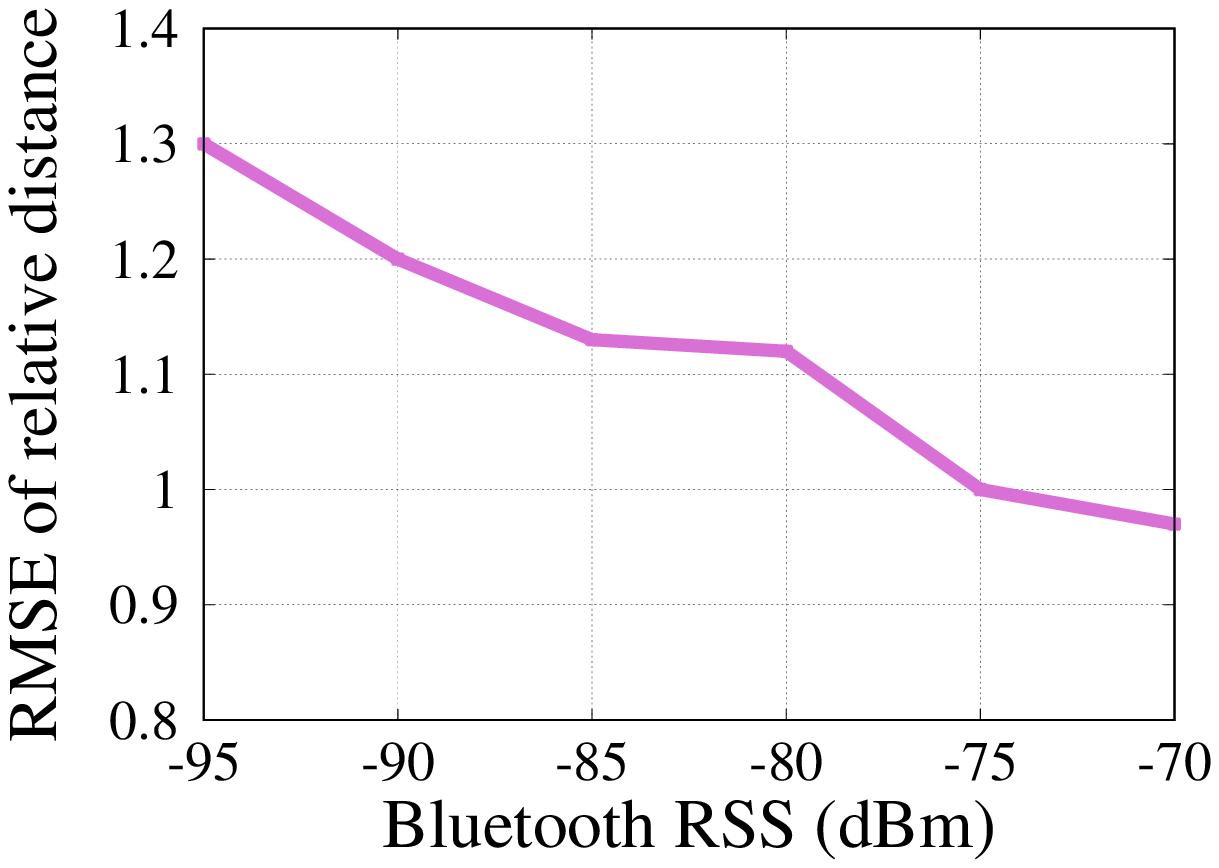}
  \caption{Relative distance error versus different Bluetooth RSSs.}
  \label{fig:bluetooth}
\endminipage\hfill
\minipage{0.23\textwidth}
  \includegraphics[width=\linewidth]{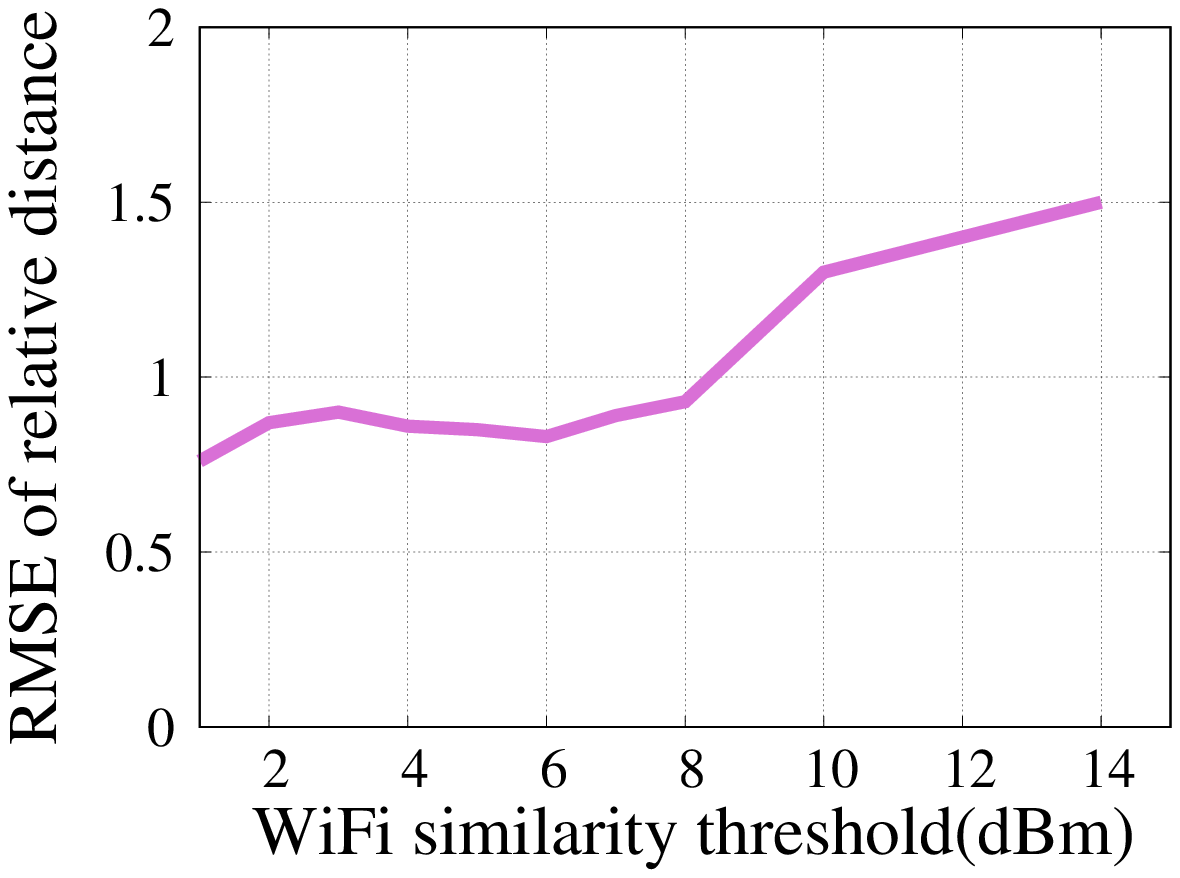}
  \caption{Relative distance error versus WiFi similarity threshold.}
  \label{fig:wifi}
\endminipage\hfill
\minipage{0.23\textwidth}
\includegraphics[width=\linewidth]{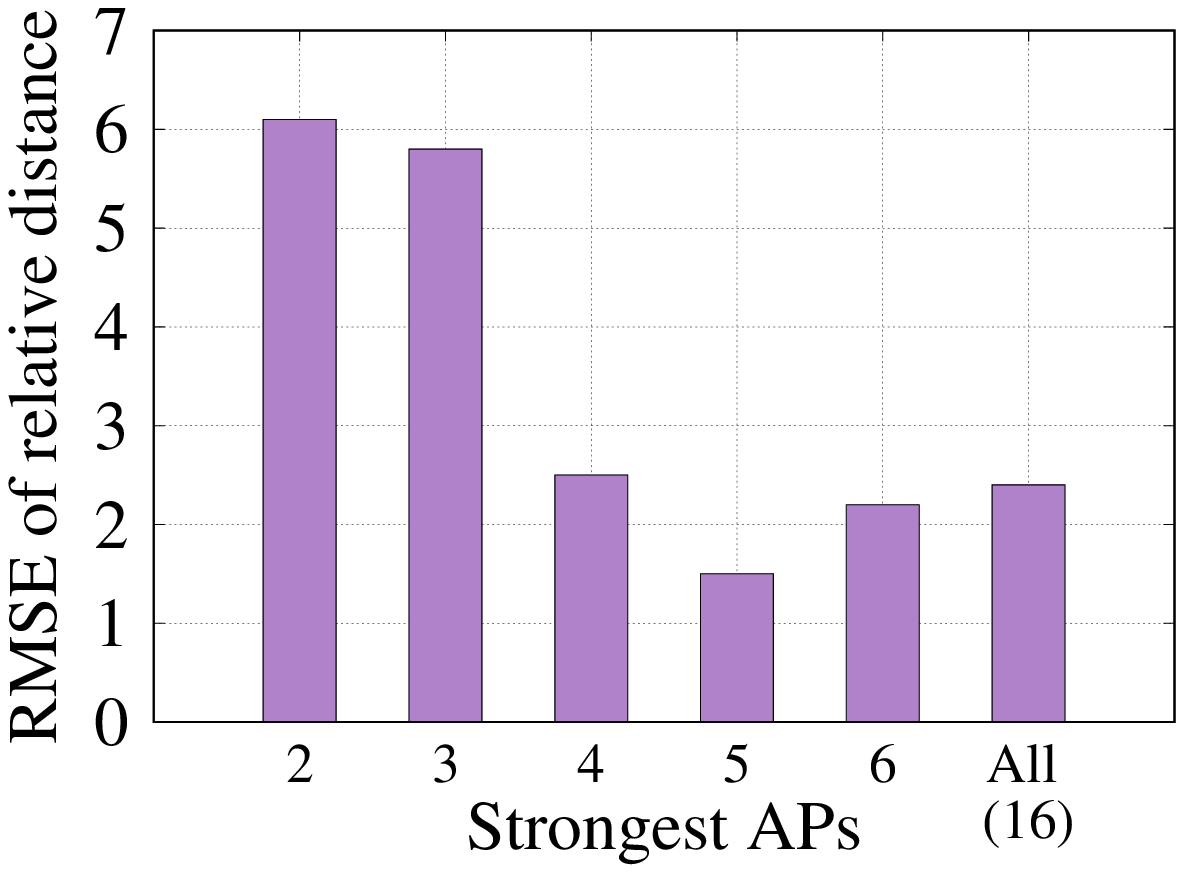}
\caption{Relative distance error using different number of WiFi APs.}
\label{fig:ap}
\endminipage\hfill
\minipage{0.23\textwidth}
  \includegraphics[width=\linewidth]{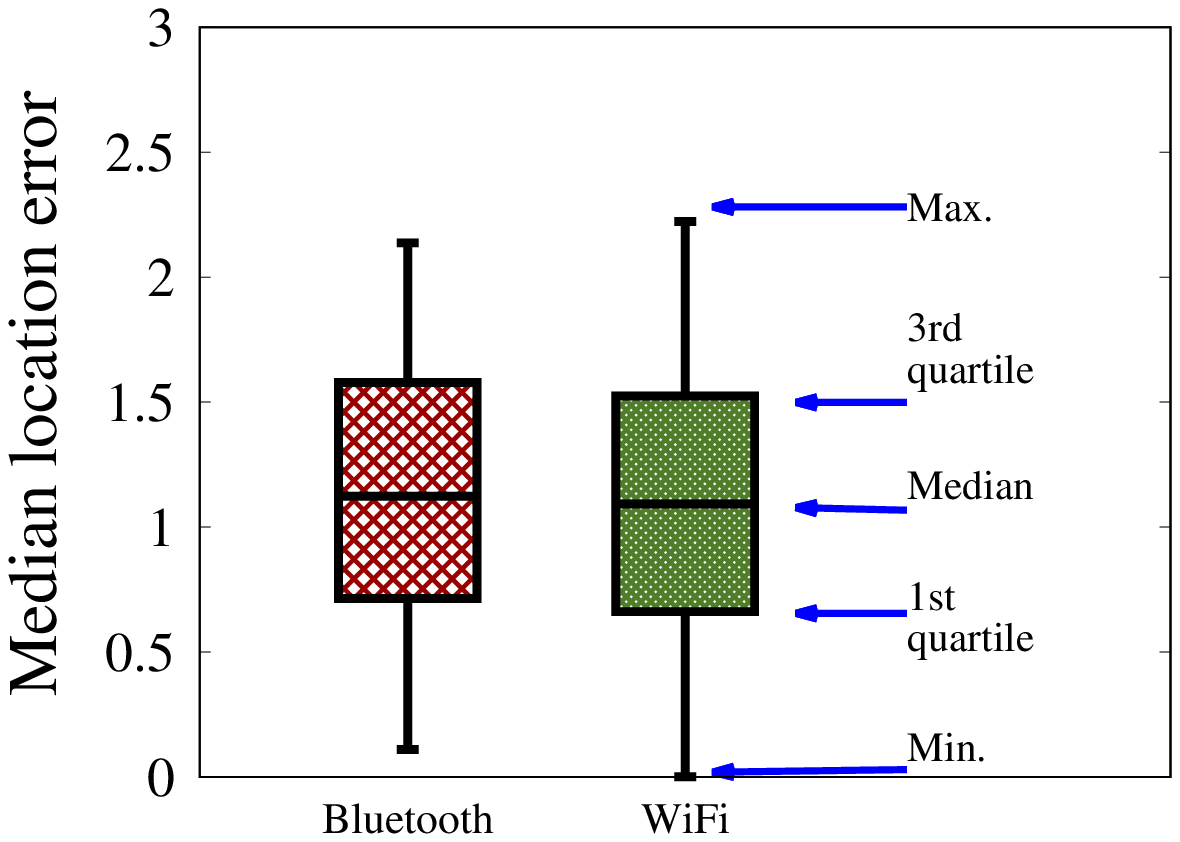}
  \caption{Effect of relative distance observation models on the location error.}
  \label{fig:compare}
\endminipage\hfill
\end{figure*}
\\
\textbf{WiFi Model}
Similarly, decreasing the similarity threshold for the WiFi signal leads to better accuracy at the expense of a shorter range (Figure~\ref{fig:wifi}). We use a WiFi threshold that is equal to 8dBm in our experiments. Figure~\ref{fig:ap} further shows the effect of controlling the number of APs used in the similarity calculation. The figure shows that increasing the number of the strongest APs used in similarity calculations leads to higher accuracy until we reach an optimal point at five APs. This can be explained by noting that a few number of APs leads to ambiguity in the user location while a large number of APs adds noise to the APs vector, both lead to reduced accuracy.

\noindent\textbf{Models Comparison}
 Figure \ref{fig:compare} compares the box-plot of relative distance error estimation by the two encounter observation models. Evident from the figure, the Bluetooth and WiFi models have comparable performance. This is an interesting result as one should expect that Bluetooth should lead to better accuracy due to its shorter range. However, WiFi has more than one AP in a scan, compared to only one device for Bluetooth. This compensates for the noise due to the higher range of WiFi. 
  Finally, we note that the two observation models are independent and can be used together to provide both higher range and better accuracy.
\begin{figure}[!t]
        \centering
        \begin{subfigure}[b]{0.23\textwidth}
                \includegraphics[width=1\textwidth]{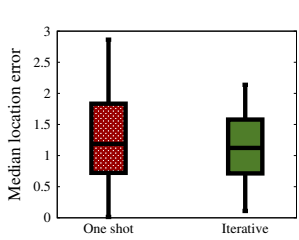}
           \caption{Localization error.}
           \label{fig:sampling_acc}
        \end{subfigure}%
        \begin{subfigure}[b]{0.23\textwidth}
                \includegraphics[width=\textwidth]{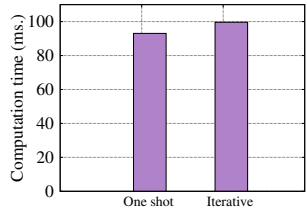}
              \caption{Computation time per estimate.}
            \label{fig:sampling_time}
        \end{subfigure}
        \caption{Comparison of system performance using one-shot sampling vs iterative sampling.}
             \label{fig:particle}
\end{figure}
\begin{figure}[!t]
        \centering
        \begin{subfigure}[b]{0.23\textwidth}
                \includegraphics[width=1\textwidth]{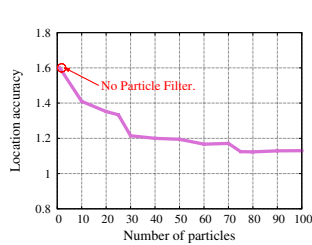}
           \caption{Localization error.}
           \label{fig:particles_acc}
        \end{subfigure}%
        \begin{subfigure}[b]{0.23\textwidth}
                \includegraphics[width=\textwidth]{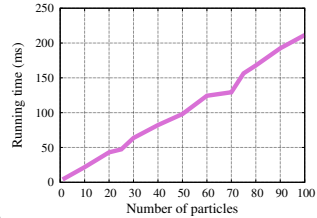}
              \caption{Computation time per estimate.}
            \label{fig:particles_time}
        \end{subfigure}
        \caption{Effect of  number of particles on system performance.}
             \label{fig:particles}
\end{figure}

\subsubsection{Sampling Method}
 Figure~\ref{fig:sampling_acc} shows the box-plot of the two sampling techniques used to update the user's location: one-shot and iterative sampling. The figure confirms that iterative sampling outperforms one-shot sampling in the median as well as the different quantile errors. This comes at the expense of increased computation time as shown in Figure \ref{fig:sampling_time}. 
\subsubsection{Number of Particles}
Figure~\ref{fig:particles_acc} shows the effect of increasing the number of particles on performance. 
The figure shows that as more particles are used to represent the user's location, the localization accuracy increases. This highlights that the particle filter does capture the dynamics of the system.
The figure also shows that as more particles are used to represent the user's location, the localization accuracy increases. Nonetheless, the computation time of the system increases linearly with the increase in the number of particles as shown in Figure \ref{fig:particles_time}. To find a comprise of the computation time- accuracy tradeoff, \sys{} uses only 75 particles. 

\begin{figure}[!t]
\centering
 \includegraphics[width=0.7\linewidth]{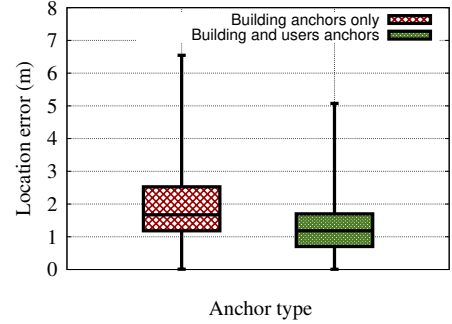}
  \caption{ Effect of using building anchors only and using users anchors on localization accuracy.}
  \label{fig:ww}
\end{figure}

\begin{figure}[!t]
\minipage{0.23\textwidth}
 \includegraphics[width=\linewidth]{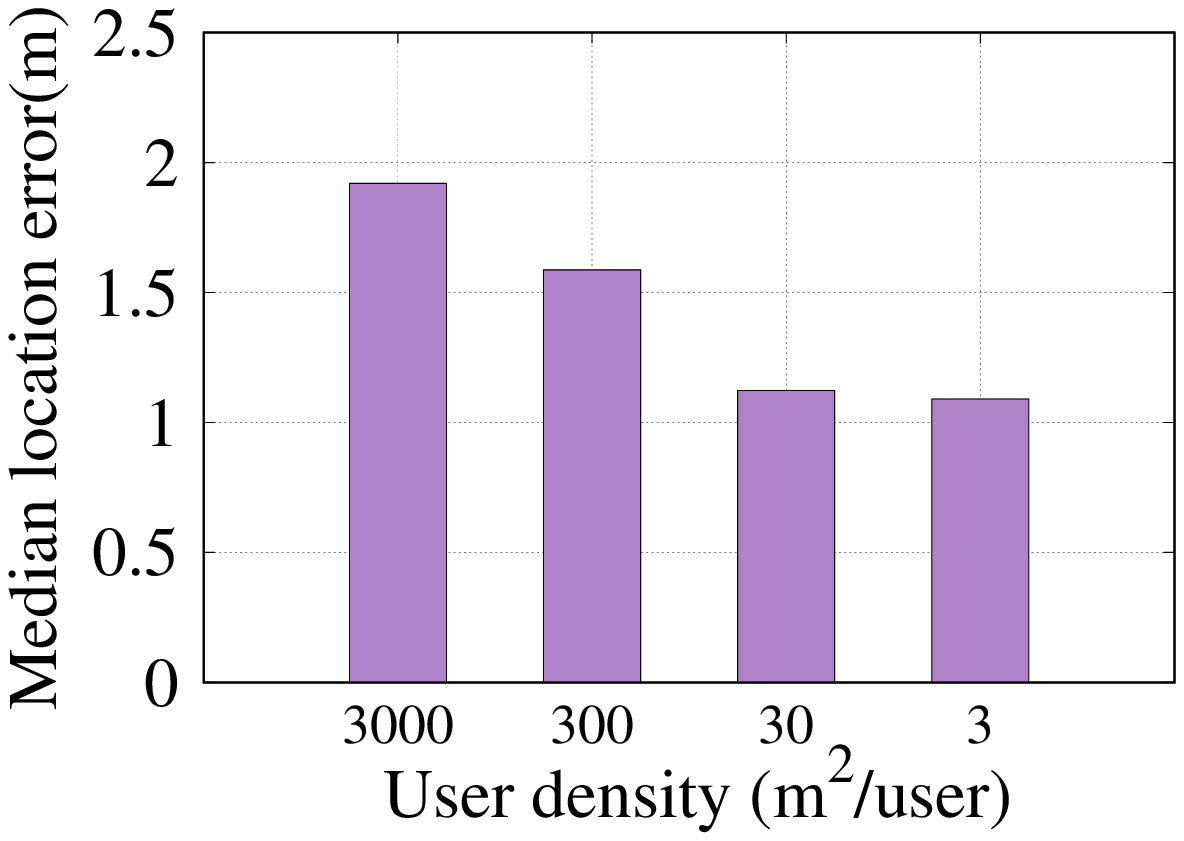}
  \caption{Effect of (\textbf{inverse}) user density on localization accuracy.}
  \label{fig:density}
\endminipage\hfill
\minipage{0.23\textwidth}
  \includegraphics[width=\linewidth]{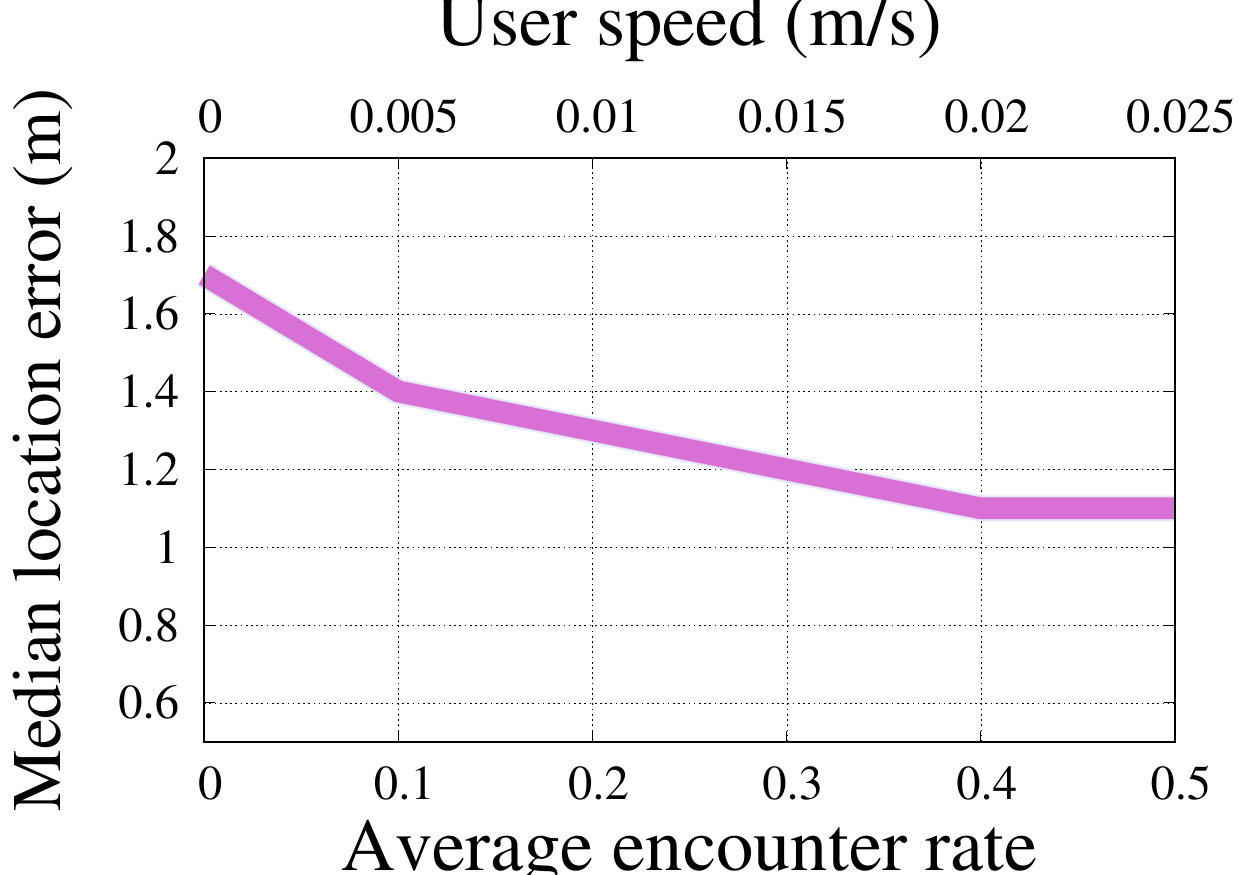}
\caption{ Effect of average encounter rate on localization accuracy.}
\label{encden}
\endminipage\hfill
\end{figure}

\begin{figure}[!t]
\minipage{0.23\textwidth}
 \includegraphics[width=\linewidth]{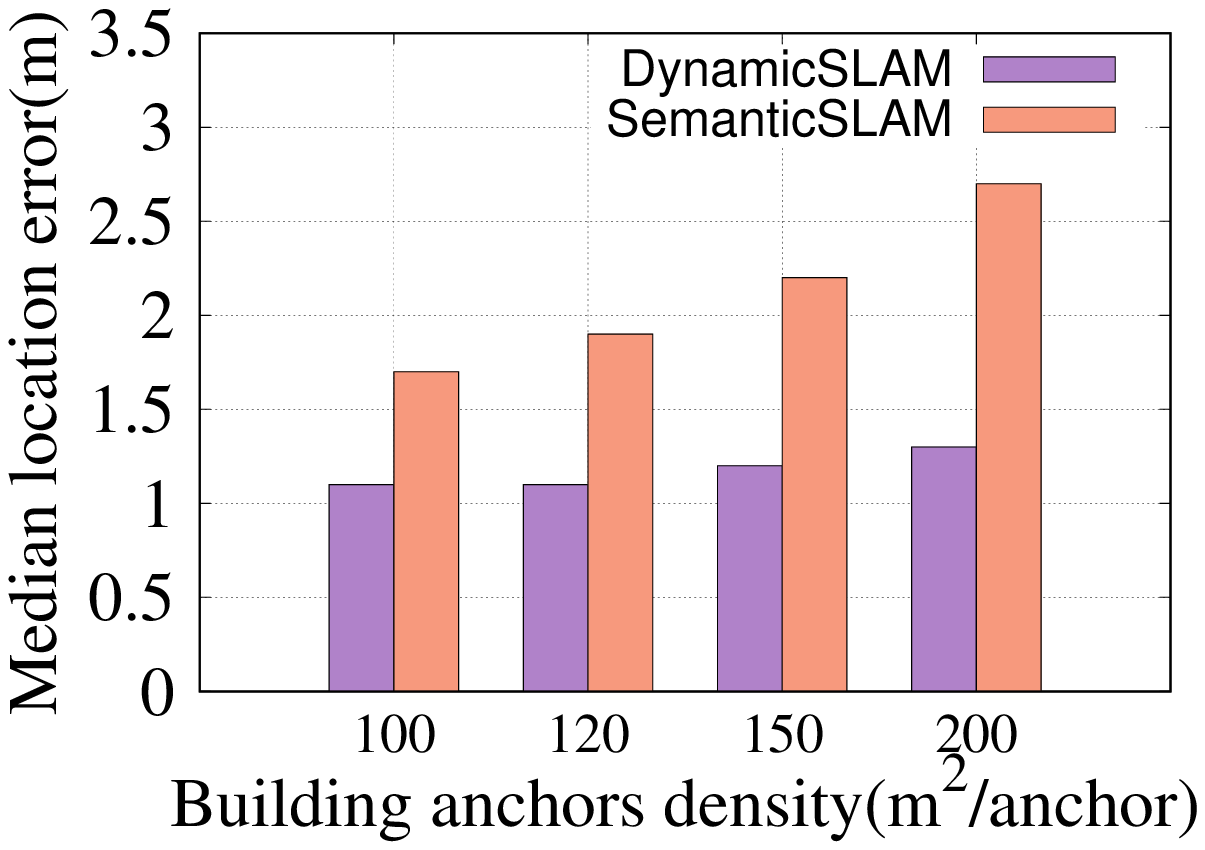}
  \caption{ Effect of (\textbf{inverse}) building anchor density on localization accuracy.}
  \label{anchdenisties}
\endminipage\hfill
\minipage{0.23\textwidth}
 \includegraphics[width=\linewidth]{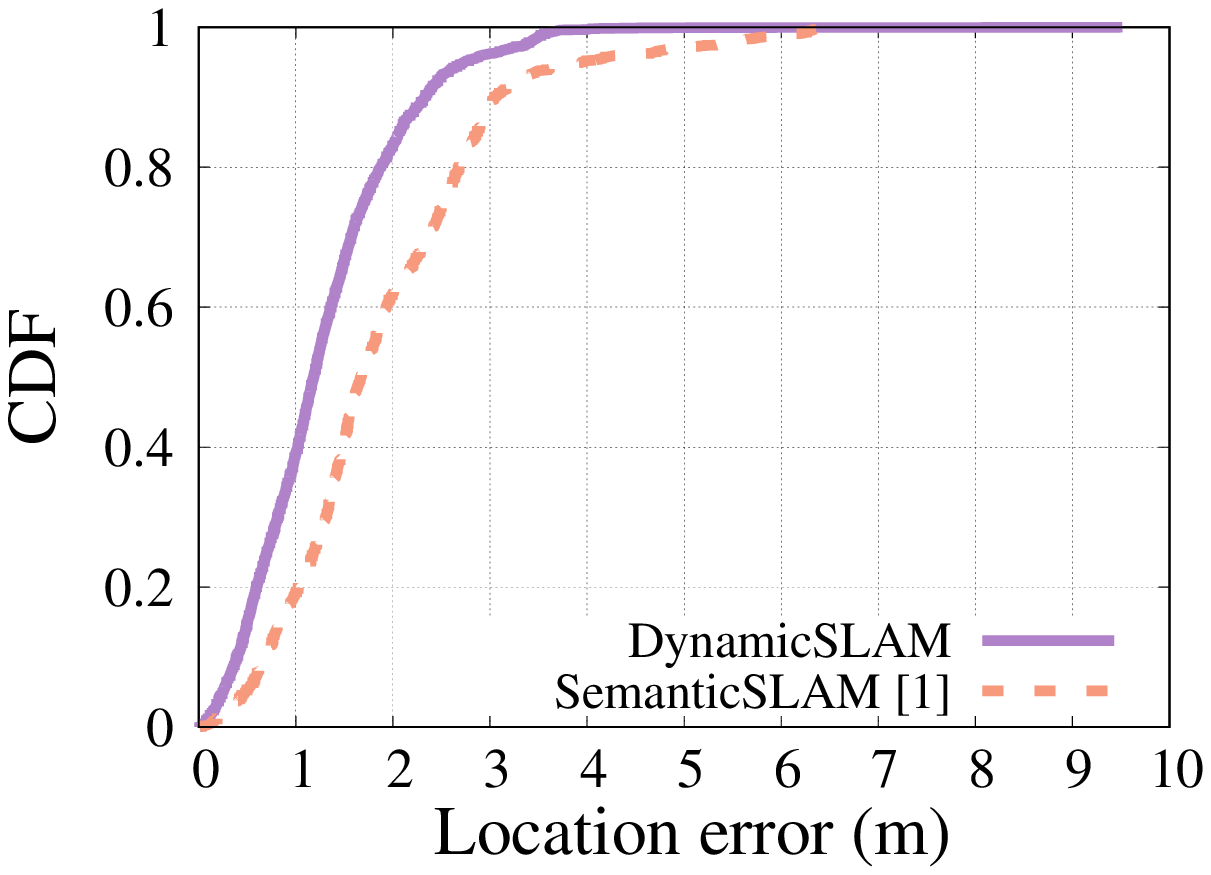}
  \caption{ Comparison between DynamicSLAM and SemanticSLAM.}
  \label{cdf}
\endminipage\hfill
\end{figure}

\subsubsection{User Density}
Figure~\ref{fig:ww} shows the effect of using the traditional building anchors only and using the new dynamic anchors on localization accuracy.
Figure \ref{fig:density} shows the effect of the number of $\text{meters}^2$ per user (\textbf{inverse} of the user density) on the median localization error. 
Figure~\ref{encden} further shows the effect of users' encounter rate on the localization accuracy.
 As
more users move around, there is a higher probability for users' encounters, leading to more opportunities of individual user location calibration, and hence better localization accuracy.
The secondary x-axis shows the user speed to achieve a certain number of encounters (linear relation) under the assumption that other users are uniformly and independently distributed in the area\cite{liu2005mobility}.

Note that, regardless of the user speed, since \sys{} needs around 150ms to get the estimated location while the typical human speed is 1.4m/s \cite{levine1999pace}; the user can be assumed stationary during the estimation process.

\subsection{Comparison With Other Systems}
In this section, we compare the localization accuracy of \sys{} against the SemanticSLAM technique that employs the SLAM approach based on the standard phone sensors \textit{using only traditional building anchors}. 
Figure \ref{anchdenisties} shows the comparison between our proposed \sys{} system and the SemanticSLAM system under different building anchors densities. The figure shows that reducing the building anchor densities affects the SemanticSLAM system accuracy while \sys{} maintains its accuracy. This is due to the fact that \sys{} leverages the users in the environment as anchors. Hence, its performance is less affected by the building anchors density.

Figure \ref{cdf} also shows the CDF of the localization accuracy for 
both systems and Table~\ref{tab:comp} summarizes the results. 
  \sys{} achieves a median distance error around 1.1m, which is better than SemanticSLAM by 55\%. Moreover, \sys{} reduces the worst case error by 29\% and enhances all other quantiles. This is due to the new human-based mobile anchors that provide ample opportunities for error resetting as illustrated in Fig~\ref{instant}.
\setlength{\belowcaptionskip}{0.1pt}
\begin{figure}[!t]
\minipage{0.23\textwidth}
 \includegraphics[width=\linewidth]{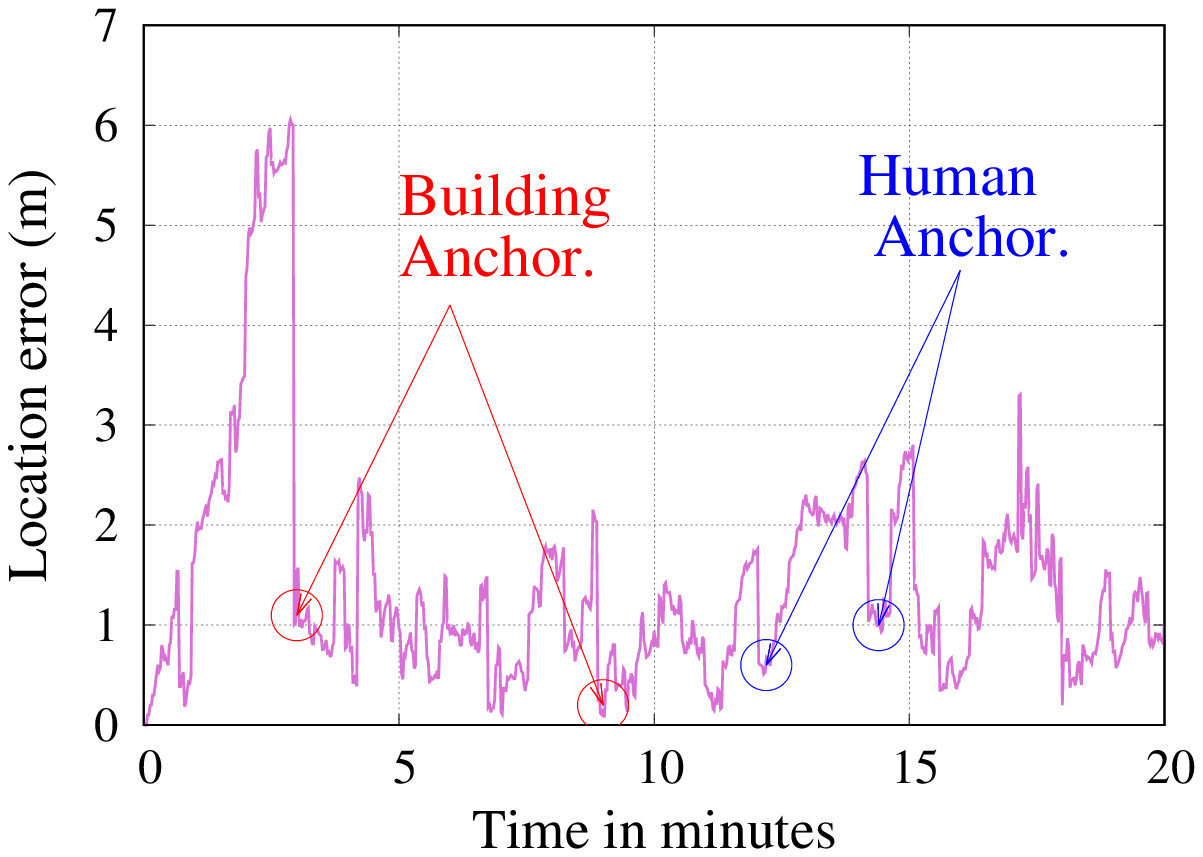}
  \caption{An example of accuracy over time.}
  \label{instant}
\endminipage\hfill
\minipage{0.26\textwidth}
  \includegraphics[width=\linewidth]{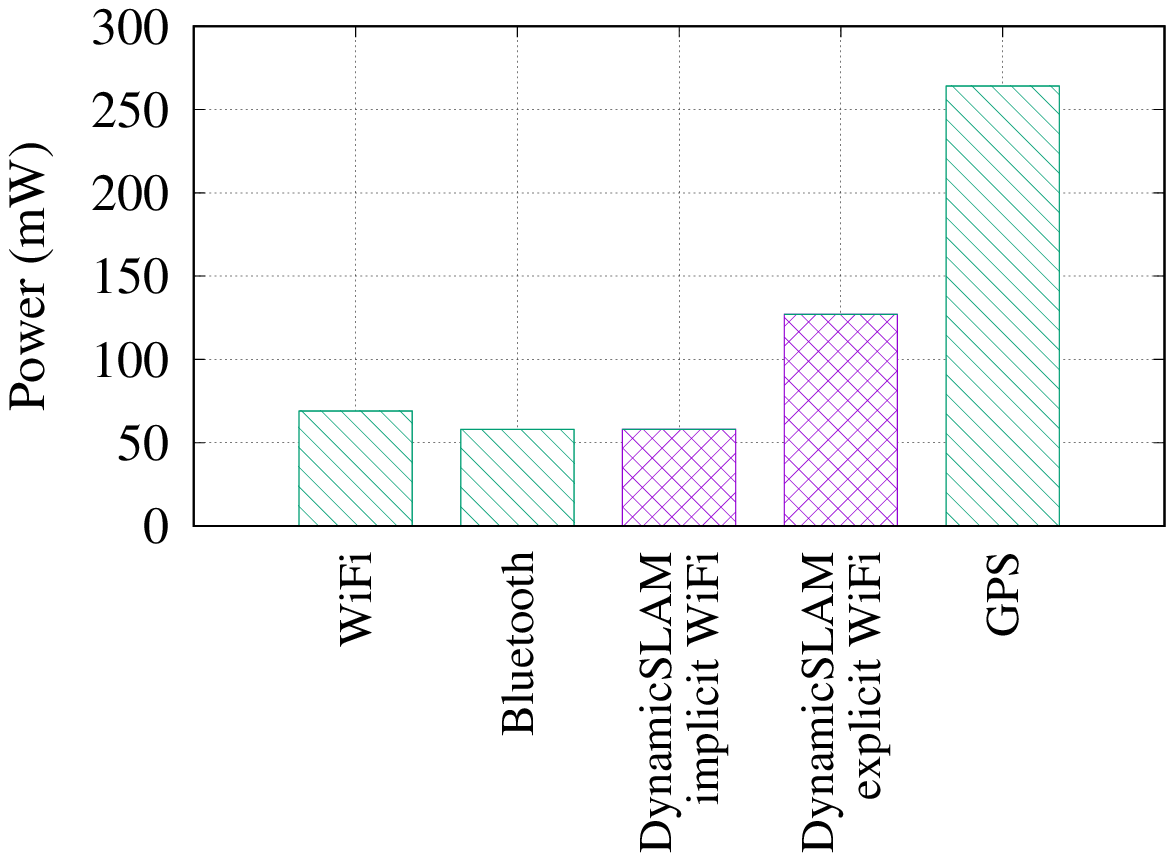}
\caption{Power consumption.}
\label{fig:power}
\endminipage\hfill
\end{figure}
\begin{table}[!t]
\centering
  \caption{Summary of the location accuracy percentiles of \sys{} vs SemanticSLAM.}
   \label{tab:comp}
\scalebox{0.73}{
\begin{tabular}{|l||l|l|l|l|l|}\hline
   Technique & \pbox{10cm}{$25^{th}$\\percentile }& \pbox{10cm}{$50^{th}$\\percentile}&\pbox{10cm}{$75^{th}$\\percentile}&\pbox{10cm}{$90^{th}$\\percentile}&Maximum\\ \hline \hline
       \pbox{10cm}{\sys{}}&0.7& \textbf{1.1}&1.7&2.3 &5.1 \\ \hline
    \pbox{10cm}{SemanticSLAM~\cite{abdelnasser2016semanticslam}}& 1.2 (-170\% )&\textbf{1.7} (-54\%)&2.6 (-52\%)&3.0 (-30\%)&6.6 (-29\%)\\ \hline
  \end{tabular}
  }
\end{table}
\subsection{Energy Footprint}
\sys{} harnesses inertial sensors, WiFi, and Bluetooth for its operation. The inertial sensors are always running to detect the phone orientation are used in a variety of mobile applications. Therefore, using them consumes virtually zero extra energy. Similarly, many users turn-on WiFi indoors to connect the Internet. Thus, as many sensors are activated by default for other purposes, this curbs the overall power consumption of \sys{}. Moreover, to further enhance \sys{} energy footprint, we adopt an adaptive sensing strategy where the sensors sampling rate and their activation are based on the nearby 
anchor density. For example,
if a certain area with high anchor density, \sys{} can decrease the sensor sampling rate as the location error will be frequently calibrated by these anchors. Additionally, Bluetooth and WiFi  can be turned off if the server predicts that the probability of different users' encounter is small.
Figure~\ref{fig:power} compares the power consumption (calculated using the PowerTutor \cite{yang2012powertutor}) of \sys{} against the GPS (at 2Hz) (used as a baseline).

\section{Related Work}
\label{sec:related}
\subsection{Indoor Localization}
Mobile phone localization has been well-studied in literature \cite{shokry2017tale,park2010growing,jin2011robust,wang2012no,shokry2018deeploc,rizk2018cellindeep}. The most ubiquitous indoor localization techniques are either WiFi-based or dead-reckoning based. WiFi-based techniques \cite{bahl2000radar,lamarca2005place}, require calibration to create a prior ``RF map'' for the building to counter the wireless channel noise. This calibration process is time consuming, tedious, and requires periodic updates. This led to the emergence of new calibration-free techniques \cite{elbakly2016robust,park2010growing}. These techniques, however, suffer from low location accuracy or requiring explicit user intervention. 
In dead-reckoning based system \cite{jin2011robust, wang2012no}, the inertial sensors on mobile phones are leveraged to dead-reckon the user starting from a reference point. Nevertheless, dead-reckoning error quickly accumulates with time; leading to a complete deviation from the actual path. Therefore, many techniques have been proposed to reset the dead-reckoning error including snapping to environment anchor points, such as elevators and stairs \cite{wang2012no, abdelnasser2016semanticslam} and matching with the map information \cite{rai2012zee}. \sys{} combines dead-reckoning as its motion model with building anchors and  its novel human-based anchors, as observation models,  that are ubiquitously in any building gives it an ample opportunities to reset the dead-reckoning error. 
\subsection{Encounter-based Localization}
Recent work has explored more social aspects of indoor localization named encounters-based approaches \cite{jun2013social,martin2014social}. 
The basic idea behind these approaches is to leverage the useful proximity information between users to enhance the localization accuracy.
SocialLoc \cite{jun2013social} utilises users' meeting or missing information to correct localization errors. A user may
intersects with another user, which is defined as an encounter. If the user does not meet another specific user during the localization process, such an event is defined as non-encounter. SocialLoc utilizes these encounter and non-encounter events to cooperatively correct the localization errors.
SocialSpring \cite{martin2014social} is a general system that sits on top of an indoor positioning system. It leverages the pairwise distance measurements between the different users in the environment to refine position estimates. 

\sys{}, on the other hand, combines building anchors with human-based anchors (users' encounter) in a collaborative SLAM framework to enhance the localization accuracy.
\subsection{Simultaneous Localization and Mapping}
 SLAM \cite{durrant2006simultaneous} is a well-known technique in the robotics domain where robot motion and observation models (e.g., odometers, laser, ultrasound) are cost-prohibitive for smartphones.
Recently, a number of SLAM-based indoor localization systems have been proposed for tracking human users. For instance, \textit{WiFi-SLAM} \cite{ferris2007wifi} builds a WiFi map based on a Gaussian Process Latent Variable Model; which allows for detecting the latent-space locations of unlabeled
signal strength data. \textit{ActionSLAM} \cite{hardegger2012actionslam} leverages a dedicated foot-mounted inertial sensor to track the user's motion while using actions uniquely related to certain locations as anchors for the user's location correction. Recently, \textit{SemanticSLAM} \cite{abdelnasser2016semanticslam}, generalizes the SLAM concept to work with smartphones without using any external hardware. It harnesses smartphone inertial sensors to dead-reckon the user's location and leverages anchor points (elevator, escalator, stairs), which can be uniquely detected by phone sensors for the accumulated error resetting. These systems, however, either require special hardware or are dependent on a limited set of anchors that are not ubiquitously available in all buildings.
\textit{\sys{}, on the other hand}, combines building anchors with the newly introduced human-based anchors that are ubiquitously in any environment.
\subsection{Multi-Robot SLAM}
Multi-Robot SLAM is an extension for SLAM that can benefit from having multiple robots in many places to carry out diverse tasks at the same time
\cite{howard2006multi,williams2002towards
,fenwick2002cooperative,birk2006merging,thrun2005multi,
ko2003practical,zhou2006multi,howard2004multi}. 
The typical ultimate target of Multi-Robot SLAM systems is to construct the joint (common) map of an unknown environment, i.e. a virtual map of the anchors the robots see. For this, they use the \textit{traditional/static building anchors} to reset the robots' accumulated error. For instance the work in \cite{howard2006multi}, uses robots' specific sensors (laser range-finders) to detect robot pairs encounter. Then, by determining their relative position using laser sensors, the system combines the observations from both robots into the common map. Hence, it \textit{indirectly} refines robots' locations by obtaining a more accurate common map. 

\sys{}, on the other hand, does enhance the environment map (traditional anchors) using a similar approach to multi-robot SLAM. In addition, we also leverage the other users (other robots in multi-robot SLAM) as a new type of anchors. Specifically, we leverage other users (robots) as a temporary anchor to \textit{directly} refine users' positions and reset accumulated error. This leads to increased anchor density and hence better localization accuracy as quantified in the evaluation section.

\section{Conclusions}
\label{sec:conc}
We presented \sys{}: an accurate, low-overhead, and collaborative simultaneous localization and mapping system. \sys{} tracks the user's location using dead-reckoning and leverages physical anchors existing in buildings as well a novel human-based dynamic anchors to reset the accumulated localization error in a unified probabilistic framework. A number of methods to detect users' encounters and to handle their relative distance measurement error have also been proposed. Furthermore, we presented a theoretical proof of system convergence as well as the human anchors ability to reset the accumulated error.

Evaluation of \sys{} in a typical university building using different
Android phones shows that it can achieve a median distance error
of 1.1m, which outperforms other state-of-the-art indoor localization techniques by approximately 55\%, highlighting its promise for ubiquitous indoor localization.

\ifCLASSOPTIONcompsoc
\else
\fi


\ifCLASSOPTIONcaptionsoff
  \newpage
\fi

\bibliographystyle{abbrv}
{\footnotesize
\bibliography{refs2}}

\begin{IEEEbiography}
[{\includegraphics[width=1in,height=2in,clip,keepaspectratio]{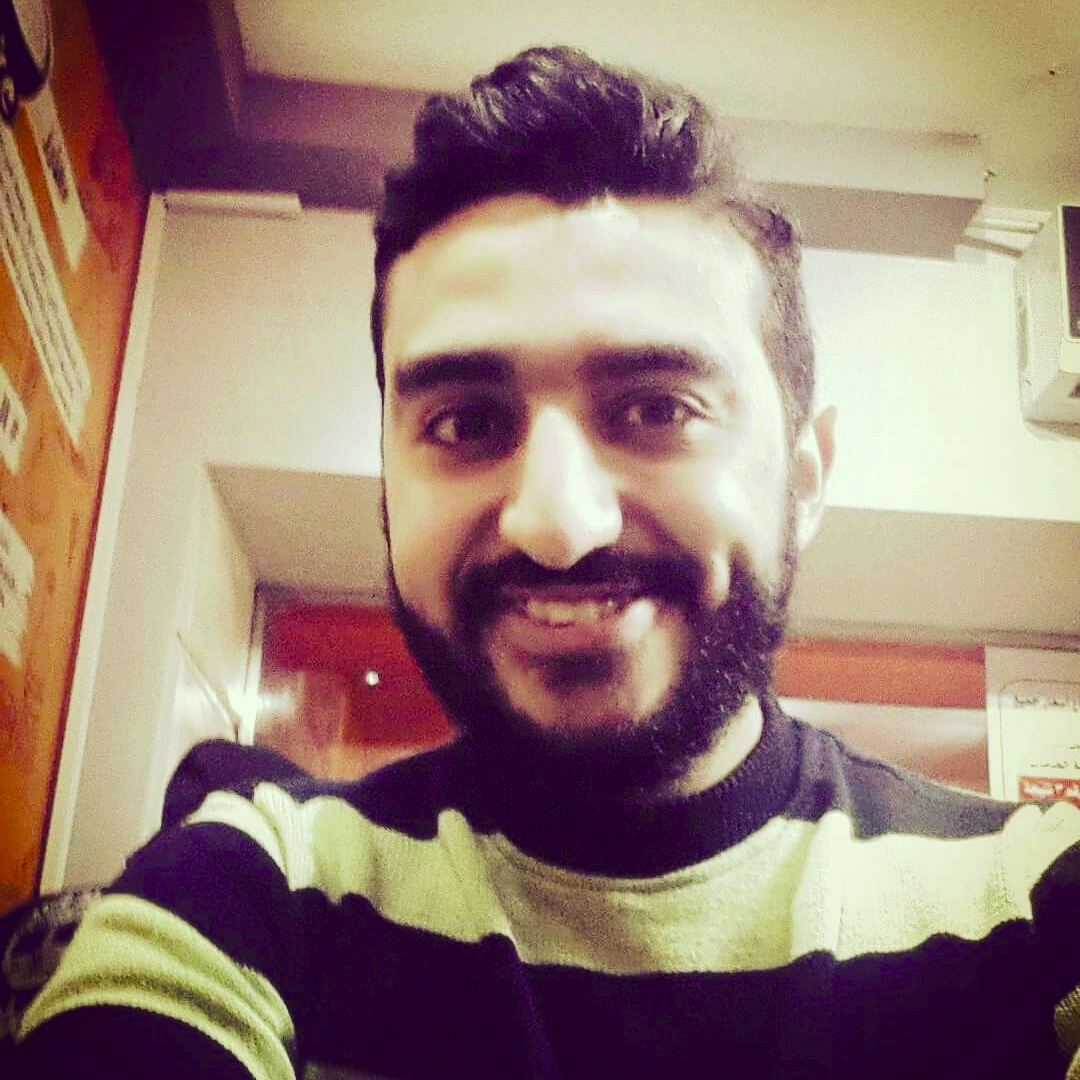}}]{Ahmed Shokry}
is with the Wireless Research Center at Alexandria University, Egypt. He received his B.Sc. in Computer and Systems
Engineering from Alexandria University in 2013. He also received his M.Sc. in Computer Science and Engineering from the same university in 2019. 
His research interests include location determination systems, theoretical and experimental algorithms.
\end{IEEEbiography}

\begin{IEEEbiography}
[{\includegraphics[width=1in,height=2in,clip,keepaspectratio]{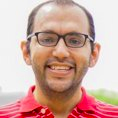}}]{Moustafa Elhamshary}
is a late assistant professor
at Tanta University, Egypt. Previously, he was a
research fellow at Graduate School of Information Science and Technology, Osaka University,
Japan. His research interests include location
determination technologies, mobile sensing, human activity recognition and location-based social networks
\end{IEEEbiography}

\begin{IEEEbiography}    [{\includegraphics[width=1in,height=2in,clip,keepaspectratio]{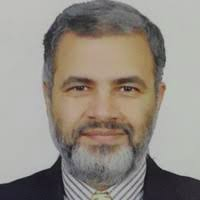}}]{Moustafa Youssef}
is a professor at Alexandria
University and founder \& director of the Wireless
Research Center of Excellence, Egypt. His research
interests include mobile wireless networks, mobile
computing, location determination technologies, pervasive computing, and network security. He has tens
of issued and pending patents. He is the Lead Guest
Editor of the upcoming IEEE Computer Special
Issue on Transformative Technologies, an Associate
Editor for IEEE TMC and ACM TSAS, an Area
Editor of Elsevier PMC, as well as served on the
organizing and technical committees of numerous prestigious conferences.
He is the recipient of the 2003 University of Maryland Invention of the
Year award, the 2010 TWAS-AAS-Microsoft Award for Young Scientists,
the 2013 and 2014 COMESA Innovation Award, the 2013 ACM SIGSpatial
GIS Conference Best Paper Award, the 2017 Egyptian State Award, multiple
Google Research Awards, among many others. He is an ACM Distinguished
Speaker, an ACM Distinguished Scientist, and an IEEE Fellow.
\end{IEEEbiography}



\end{document}